\pdfoutput=1

\newif\iffull
\fulltrue

\iffull
  \documentclass[11pt]{article}
  \usepackage{fullpage}
\else
  \documentclass[conference,compsoc]{IEEEtran}
\fi

\iffull
  \newif\ifCLASSOPTIONcompsoc
\fi

\iffull
  \usepackage{cite}
\else
  \ifCLASSOPTIONcompsoc
    \usepackage[nocompress]{cite}
  \else
    \usepackage{cite}
  \fi

\fi

\usepackage[T1]{fontenc}
\usepackage[utf8]{inputenc}
\usepackage{microtype}

\usepackage{amsmath, amssymb, amsthm}
\usepackage{mathtools}
\usepackage{graphicx}
\usepackage{pgfplots}
\pgfplotsset{compat=1.18}
\usepackage{booktabs}
\usepackage{multirow}
\usepackage{array}

\usepackage{algorithm}
\usepackage{algpseudocode}

\usepackage{placeins}

\usepackage{xcolor}
\usepackage{xspace}
\usepackage[
  colorlinks=true,
  linkcolor={blue!55!black},
  citecolor={blue!55!black},
  urlcolor={blue!55!black},
  breaklinks=true,
]{hyperref}
\usepackage[capitalize]{cleveref}
\usepackage{lastpage}

\iffull\else
\makeatletter
\renewcommand\paragraph{%
  \@startsection{paragraph}{4}{\z@}%
    {1.5ex \@plus 0.5ex \@minus 0.2ex}
    {-0.5em}
    {\normalfont\normalsize\bfseries\def\@IEEEsectpunct{\ }}}
\makeatother
\fi

\usepackage{aliascnt}
\newcommand{\newaliasedtheorem}[3][theorem]{%
  \newaliascnt{#2}{#1}%
  \newtheorem{#2}[#2]{#3}%
  \aliascntresetthe{#2}%
  \crefname{#2}{#3}{#3s}%
  \Crefname{#2}{#3}{#3s}%
}

\newtheorem{theorem}{Theorem}
\crefname{theorem}{Theorem}{Theorems}
\Crefname{theorem}{Theorem}{Theorems}
\newaliasedtheorem{lemma}{Lemma}
\newaliasedtheorem{proposition}{Proposition}
\newaliasedtheorem{corollary}{Corollary}
\newaliasedtheorem{claim}{Claim}
\newaliasedtheorem{fact}{Fact}
\theoremstyle{definition}
\newaliasedtheorem{definition}{Definition}
\newaliasedtheorem{construction}{Construction}
\theoremstyle{remark}
\newaliasedtheorem{remark}{Remark}

\newif\ifnotes
\notestrue
\ifnotes
  \newcommand{\rnote}[1]{{\color{red} {Ron: #1}}}
  \newcommand{\cnote}[1]{{\color{blue} {Claude: #1}}}
  \newcommand{\bnote}[1]{{\color{green} {Benedikt: #1}}}
  \newcommand{\william}[1]{{\color{purple} {William: #1}}}
\else
  \newcommand{\rnote}[1]{}
  \newcommand{\cnote}[1]{}
  \newcommand{\bnote}[1]{}
  \newcommand{\william}[1]{}
\fi

\newcommand{\sysname}{Flock\xspace}
\newcommand{\N}{\mathbb{N}}
\newcommand{\F}{\mathbb{F}}
\newcommand{\Ftwo}{\F_2}
\newcommand{\Feight}{\F_{2^8}}
\newcommand{\Fbig}{\F_{2^{128}}}
\def\GF(#1){\F_{#1}}
\newcommand{\eq}{\mathrm{eq}}
\newcommand{\bitset}{\{0,1\}}

\makeatletter
\newcommand{\benchdef}[2]{\expandafter\def\csname benchval@#1\endcsname{#2}}
\newcommand{\bench}[1]{%
  \@ifundefined{benchval@#1}%
    {\PackageWarning{flock-bench}{unknown benchmark key `#1'}\textbf{??#1??}}%
    {\csname benchval@#1\endcsname}}
\makeatother

\benchdef{meta.hardware}            {Apple M4 Max}
\benchdef{meta.ram}                 {36}
\benchdef{meta.cores}               {10}
\benchdef{meta.threads}             {10}
\benchdef{meta.security}            {100}
\benchdef{kc.native}                {7.51}
\benchdef{kc.flock.st}              {30.7}
\benchdef{kc.flock.st.exp}          {14}
\benchdef{kc.flock.mt}              {252}
\benchdef{kc.flock.mt.exp}          {16}
\benchdef{kc.hc.st}                 {17.2}
\benchdef{kc.hc.st.exp}             {12}
\benchdef{kc.hc.mt}                 {66.2}
\benchdef{kc.hc.mt.exp}             {18}
\benchdef{kc.b64.st}                {5.0}
\benchdef{kc.b64.st.exp}            {12}
\benchdef{kc.b64.mt}                {37.6}
\benchdef{kc.b64.mt.exp}            {14}
\benchdef{kc.p3.st}                 {0.8}
\benchdef{kc.p3.st.exp}             {10}
\benchdef{kc.p3.mt}                 {6.3}
\benchdef{kc.p3.mt.exp}             {12}
\benchdef{overhead.kc}              {245}
\benchdef{kc.flock.st.head}         {30}
\benchdef{sha.native}               {7.14}
\benchdef{sha.flock.st}             {42.1}
\benchdef{sha.flock.st.exp}         {14}
\benchdef{sha.flock.mt}             {338}
\benchdef{sha.flock.mt.exp}         {16}
\benchdef{sha.b64.st}               {5.0}
\benchdef{sha.b64.st.exp}           {12}
\benchdef{sha.b64.mt}               {36.0}
\benchdef{sha.b64.mt.exp}           {13}
\benchdef{sha.nn.st}                {0.1}
\benchdef{sha.nn.st.exp}            {9}
\benchdef{sha.nn.mt}                {0.6}
\benchdef{sha.nn.mt.exp}            {9}
\benchdef{overhead.sha}             {170}
\benchdef{sha.flock.st.head}        {42}
\benchdef{bl.native}                {18.4}
\benchdef{bl.flock.st}              {82.1}
\benchdef{bl.flock.st.exp}          {16}
\benchdef{bl.flock.mt}              {661}
\benchdef{bl.flock.mt.exp}          {18}
\benchdef{bl.b64.st}                {6.1}
\benchdef{bl.b64.st.exp}            {12}
\benchdef{bl.b64.mt}                {47.1}
\benchdef{bl.b64.mt.exp}            {14}
\benchdef{bl.p3.st}                 {5.9}
\benchdef{bl.p3.st.exp}             {12}
\benchdef{bl.p3.mt}                 {45.9}
\benchdef{bl.p3.mt.exp}             {14}
\benchdef{overhead.bl}              {224}
\benchdef{bl.flock.st.head}         {82}
\benchdef{bl.flock.mt.head}         {660}
\benchdef{sp.kc.mt.hc}              {4}
\benchdef{sp.kc.mt.b64}             {7}
\benchdef{sp.kc.mt.p3}              {40}
\benchdef{sp.kc.st.hc}              {1.8}
\benchdef{sp.sha.mt.b64}            {9.4}
\benchdef{sp.sha.mt.b64.head}       {9}
\benchdef{sp.sha.mt.nn}             {595}
\benchdef{sp.sha.st.nn.h}           {400}
\benchdef{sp.sha.mt.nn.h}           {500}
\benchdef{sp.bl.mt.head}            {14}
\benchdef{stwo.bl2s.st}             {2.9}
\benchdef{stwo.bl3.st}              {4.1}
\benchdef{sp.bl.st.stwo}            {20}
\benchdef{range.p3.lo}              {10}
\benchdef{range.p3.hi}              {40}
\benchdef{kcfix.exp}                {14}
\benchdef{kcfix.flock.fast.thru}    {239}
\benchdef{kcfix.flock.fast.proof}   {386}
\benchdef{kcfix.flock.fast.verify}  {5.8}
\benchdef{kcfix.flock.slim.thru}    {189}
\benchdef{kcfix.flock.slim.proof}   {200}
\benchdef{kcfix.flock.slim.verify}  {5.3}
\benchdef{kcfix.flock.sec.thru}     {241}
\benchdef{kcfix.flock.sec.proof}    {558}
\benchdef{kcfix.flock.sec.verify}   {6.1}
\benchdef{kcfix.hc.thru}            {42.1}
\benchdef{kcfix.hc.proof}           {664}
\benchdef{kcfix.hc.verify}          {28}
\benchdef{kcfix.b64.thru}           {38.0}
\benchdef{kcfix.b64.proof}          {457}
\benchdef{kcfix.b64.verify}         {667}
\benchdef{kcfix.p3.thru}            {5.8}
\benchdef{kcfix.p3.proof}           {3.35}
\benchdef{kcfix.p3.verify}          {20.5}
\benchdef{kcfix.flock.verify.approx}{6}
\benchdef{phase.kc.witness_gen}     {8.0}
\benchdef{phase.kc.pcs_commit}      {26.4}
\benchdef{phase.kc.zerocheck}       {42.3}
\benchdef{phase.kc.lincheck}        {6.4}
\benchdef{phase.kc.pcs_open}        {16.8}
\benchdef{phase.sha.witness_gen}    {10.1}
\benchdef{phase.sha.pcs_commit}     {26.1}
\benchdef{phase.sha.zerocheck}      {38.6}
\benchdef{phase.sha.lincheck}       {5.9}
\benchdef{phase.sha.pcs_open}       {19.3}
\benchdef{phase.bl.witness_gen}     {5.0}
\benchdef{phase.bl.pcs_commit}      {24.4}
\benchdef{phase.bl.zerocheck}       {37.3}
\benchdef{phase.bl.lincheck}        {9.1}
\benchdef{phase.bl.pcs_open}        {24.3}
\benchdef{phase.pcszero.lo}         {84}
\benchdef{phase.pcszero.hi}         {86}
\benchdef{scaling.ymax.st}          {375}
\benchdef{scaling.ymax.mt}          {750}
\benchdef{scaling.speedup.max}      {8}
\benchdef{const.hashes_per_sig}     {160}
\benchdef{const.vitalik}            {200}
\benchdef{const.btc.sigs}           {5{,}000}
\benchdef{app.eth.tps}              {4{,}000}
\benchdef{app.btc.seconds}          {2.5}
\benchdef{head.overhead}            {250}
\benchdef{head.proof.kb}            {450}
\benchdef{head.proof.gates}         {30}

\newcommand{\blakeScalingST}{(10,46.908) (12,70.89) (14,79.468) (16,82.068) (18,80.66)}
\newcommand{\blakeScalingMT}{(10,185.172) (12,339.635) (14,528.005) (16,638.068) (18,660.512)}
\newcommand{\blakeScalingLabels}{%
  \node[above=3pt, font=\small] at (axis cs:10,185.172) {$3.9\times$};
  \node[above=3pt, font=\small] at (axis cs:12,339.635) {$4.8\times$};
  \node[above=3pt, font=\small] at (axis cs:14,528.005) {$6.6\times$};
  \node[above=3pt, font=\small] at (axis cs:16,638.068) {$7.8\times$};
  \node[above=3pt, font=\small] at (axis cs:18,660.512) {$8.2\times$};
}

\newif\ifsubmission
\submissionfalse   

\title{\texorpdfstring{\sysname}{Flock}: Fast Proving for Batch Boolean Computations}

\iffull
  \ifsubmission
  \else
    \author{
      Benedikt B\"{u}nz\thanks{NYU, Espresso Systems. Email: \url{bb@nyu.edu}} \and
      Ron D. Rothblum\thanks{Succinct. Email: \url{rothblum@gmail.com}} \and
      William Wang\thanks{NYU. Email: \url{ww@priv.pub}}
    }
  \fi
\else
  \ifsubmission
    \author{\IEEEauthorblockN{Paper \#1367, \pageref{LastPage} pages}}
  \else
    \author{%
      \IEEEauthorblockN{First Author}
      \IEEEauthorblockA{Affiliation\\Email: first@example.com}
      \and
      \IEEEauthorblockN{Second Author}
      \IEEEauthorblockA{Affiliation\\Email: second@example.com}
    }
  \fi
\fi

\iffull
\makeatletter
\let\dtrt@startsection\@startsection
\def\@startsection#1#2#3#4#5#6{
  \def\dtrt@secttype{#1}\dtrt@startsection{#1}{#2}{#3}{#4}{#5}{#6}}
\newif\ifdtrt@inTOC
\let\dtrt@toc\tableofcontents
\renewcommand\tableofcontents{\begingroup\dtrt@inTOCtrue\dtrt@toc\endgroup}
\let\dtrt@ssect\@ssect
\renewcommand\@ssect[5]{
  \dtrt@ssect{#1}{#2}{#3}{#4}{%
    \phantomsection
    \ifdtrt@inTOC\else\addcontentsline{toc}{\dtrt@secttype}{#5}\fi
    #5}}
\makeatother
\fi

\begin{document}
\maketitle

  \pagestyle{plain}
  \thispagestyle{plain}

\iffull\else
  \IEEEpeerreviewmaketitle
\fi

\begin{abstract}
For many applications of SNARKs, a key bottleneck is proving large batches of standard cryptographic hash evaluations, such as SHA-256, Keccak, or BLAKE3. We introduce \sysname{}, a hash-based SNARK for extremely fast proving of such batched Boolean computations. \sysname{} proves batches of the same R1CS circuit (plus input/output relations between them), can prove hash-chains and Merkle path openings, and in principle can be extended to full-fledged hash-based signature verification. At its core, \sysname{} combines new optimizations for the lincheck and zerocheck protocols with an aggressively optimized proof-of-concept implementation co-designed by coding agents.

On a single core of an M4 Max processor, \sysname{} proves \bench{bl.flock.st.head}k evaluations of the BLAKE3 compression function, \bench{sha.flock.st.head}k SHA-256 compressions, and \bench{kc.flock.st.head}k Keccak permutations per second --- less than a $\bench{head.overhead}\times$ overhead over native execution. On ten cores, throughput exceeds \bench{bl.flock.mt.head}k BLAKE3 compressions per second; in proving SHA-256, \sysname{} is more than $\bench{sp.sha.mt.b64.head}\times$ faster than Binius64, the prior state of the art, and more than $\bench{sp.sha.mt.nn.h}\times$ faster than the fastest elliptic curve-based SNARK we measured against.

\end{abstract}

\iffull
  \newpage
  \setcounter{tocdepth}{2}
  \tableofcontents
  \newpage
\fi

\section{Introduction}
\label{sec:intro}

Succinct non-interactive arguments of knowledge (SNARKs) let an untrusted prover convince a verifier that a computation was executed correctly, by sending a short certificate that the verifier checks in time that is \emph{much} faster than actually performing the computation.

The key bottleneck in SNARKs is the complexity of \emph{proving} correctness of the computation. The prover's work scales with the complexity of the computation being proved, and the overhead can be quite large. State-of-the-art provers for general-purpose computations (aka zkVMs) are four to six orders of magnitude slower than native execution~\cite{Thalersblog}. While there are other important resources such as proof-size and verification time, these can usually be solved using recursive proof composition.\footnote{Recursive proof composition\cite{Valiant08,CT10}, combines a fast but \emph{largish} SNARK system with one that has short proofs but a more expensive prover. This composition delivers the best-of-both-worlds (i.e., fast provers \emph{and} short proofs) and is common in practice, see, e.g., \cite{Nova,Halo2,Plonky2}.} Thus, the prover time is the key remaining limiting factor for most applications.

This motivates the central question of this paper:
\begin{center}\itshape
  What is the smallest concrete overhead that proving can have over native execution?
\end{center}

Continuing a line of theoretical \cite{RR24,RR25,HR22,ARR25} and more recently also applied \cite{binius,biniusml,hashcaster,Irr} work, we focus on proving correctness of computations expressed as \emph{Boolean circuits} (i.e., with bit-wise AND and XOR gates). This is in contrast to the vast majority of SNARK research, which focuses on proving correctness of arithmetic computation over large finite fields.

Boolean circuits are a natural target for two reasons. First, they serve as a good proxy for the actual cost of general computation. Second, they are perfectly suited to expressing the standard cryptographic primitives that dominate key applications of SNARKs.
More generally, the same thing that makes a computation efficient natively — fewer bit operations — is what makes it cheap to prove in our system. Application designers optimize for efficiency anyway, and that optimization carries straight through to the prover.

\paragraph{\sysname: Fast proving for batch computations.} We introduce \sysname, a proof-system designed and implemented to prove \emph{batch} computations: that is, workloads consisting of many independent invocations of the same Boolean circuit, tied together by (simple) relations on their inputs and outputs. Concretely, fix a Boolean circuit $F$ and an arbitrary auxiliary input/output (IO) circuit $G$; \sysname{} proves statements of the form
\begin{gather*}
  F(\mathbf{x}_i, \mathbf{y}_i) = 0 \quad \forall\, i \in [K], \text{ and}\\
  G\bigl((\mathbf{x}_i)_{i \in [K]}, (\mathbf{y}_i)_{i \in [K]}\bigr) = 0.
\end{gather*}
The idea is for the repeated relation $F$ to capture the bulk of the work; the ``glue circuit $G$'' is in principle unrestricted but should be simple. For example, the $F$ circuits can be hash evaluations and $G$ checks that the inputs and outputs correspond to a hash-chain or a set of Merkle paths. 

The batch setting is motivated both theoretically and practically. Thaler~\cite{Thaler13} demonstrated the efficiency of batch proving using a variant of the GKR~\cite{GKR15} protocol. Proving a batch is significantly simpler than proving arbitrary circuits, as it obviates the need for expensive permutation or memory checking arguments~\cite{Setty20,HyperPlonk23,Marlin} and preprocessing. The verifier can be linear in the evaluation of one computation but sublinear in the entire batch size. For this reason, batch computation has been well studied in theory in a line of work which showed that for sufficiently uniform computations one can in principle prove ``as fast as one computes'' \cite{RR25,HR22}. 

Batch computation is also of immense practical interest: Almost every deployed SNARK workload is dominated by batched copies of a single small computation: hash chains in verifiable delay functions, Merkle openings, aggregation of hash-based signatures, and the inner layer of recursive proving are all collections of many invocations of the same primitive (with some small logic on top). Even general machine computation contains repeated steps of the same primitive (a single instruction execution). Most recently, the Ethereum Foundation specified Lean VM\cite{leanvm}, with the specific goal of supporting their post-quantum transition through SNARKs. Lean VM is a read-only memory virtual machine with native support for hashing, and workloads such as signature aggregation are dominated by the hashing instructions.

In our implementation, we focus on the batch computation of functions $F$ corresponding to hash function evaluations and a glue circuit $G$ that supports basic linear IO relations. This suffices for hash-chains, Merkle trees and independent hashing. These circuits are sufficiently rich to capture many applications, e.g. VDFs, and demonstrate the power of our proof system. Our protocol can be easily extended by implementing more complex $G$, or can be tied together with a compatible multilinear proof system over binary fields 
(e.g., simple adaptations of PIOPs such as \cite{Setty20,HyperPlonk23})

\paragraph{Why hashes, and why standard hashes.} We benchmark \sysname{} on the most ubiquitous batch workload, cryptographic hashing, and we deliberately target \emph{standard}, hardware-friendly hash functions---Keccak, SHA-256, and BLAKE3---rather than the ``SNARK-friendly'' primitives (e.g.\ Poseidon\cite{Poseidon,GKS23,GKKRSSW26}, Rescue \cite{Rescue}, MiMC \cite{MiMc}) that much of the literature co-designs with its proof systems. Standard hashes are what most real-world applications already use, from Merkle trees and verifiable delay functions to the hash-based signatures proposed for post-quantum Bitcoin and Ethereum, and they are the most heavily scrutinized primitives in deployment; SNARK-friendly designs are far less battle-tested, and their algebraic structure has repeatedly enabled attacks~\cite{AshurHADES23,SteinerPoseidonGB24,GrassiNeptune25,CheapLunch25,ZhaoGraeffeNTT25,MerzPoseidon2_26}. More fundamentally, a SNARK-friendly hash is cheap to prove only inside a proof system built over the single field it was designed for: it \emph{enshrines} that field, so only SNARKs over the same field can exploit it. This severely limits composability and upgradability --- swapping the proof system, or composing two systems over different fields, forfeits the speedup --- and it forces the application designer into a false choice: optimize the hash for the prover and pay for it in native execution, or optimize for native speed and pay for it in the prover.

Targeting a standard hash removes this choice. The hash is fixed by the application and runs identically whether or not a proof is ever produced, so the only remaining question is how cheaply it can be proved. By showing that \sysname's prover throughput essentially tracks the native speed of standard hashes --- while remaining best-in-class on what are arguably the three most popular hash functions --- we let designers pick a hash on its own merits and obtain efficient proving regardless. This also makes ``prover time vs.\ native time'' a fair comparison: standard hashes cost only a handful of cycles per byte on modern hardware, so any inefficiency in the proof system is immediately visible against a single-digit-nanoseconds-per-byte baseline. SNARK-friendly hashes, by contrast, flatter the prover by making the native baseline itself slow.

\subsection{Our results}
We design and implement \sysname, a hash-based post-quantum SNARK for batch Boolean R1CS computations. Our headline number, measured on a single core of an Apple M4 Max:
\[
  \frac{\text{prover time per hash}}{\text{native time per hash}} \;\approx\; \bench{head.overhead},
\]
where the native time measures independent \emph{computations} of the underlying permutation/compression function.\footnote{ARM has special instructions for acceleration of SHA-256 and SHA3 compressions. Since we are interested in an apples-to-apples comparison, we do not enable these instructions when computing the native execution time.}

On the same computer \sysname proves $\bench{bl.flock.st.head}$k BLAKE3 compressions, $\bench{sha.flock.st.head}$k SHA-256 compressions, and $\bench{kc.flock.st.head}$k Keccak permutations per second on a single core (this includes witness generation which benchmarks sometimes exclude). The implementation parallelizes well: on ten cores throughput exceeds $\bench{bl.flock.mt.head}$k BLAKE3 compressions per second.

On proving SHA-256, \sysname is roughly $\bench{sp.sha.mt.b64.head}\times$ faster than Binius64~\cite{Irr}, the prior state-of-the-art, and over $\bench{sp.sha.mt.nn.h}\times$ faster than NeutronNova~\cite{NeutronNova,Spartan2}. It is about $\bench{sp.kc.st.hc}\times$ faster than Hashcaster \cite{hashcaster} in proving the Keccak-f permutation, and $\bench{sp.bl.mt.head} \times$ faster than Binius64 and Plonky3~\cite{Plonky3} in proving BLAKE3. See \cref{sec:eval} for the full breakdown. \sysname's proofs are less than $\bench{head.proof.kb}$ KB for computations with $2^{\bench{head.proof.gates}}$ AND gates and take only a few milliseconds to verify.

\paragraph{From hashes to signatures.} These throughputs translate into application-level rates. Post-quantum hash-based signatures\cite{Lamport,Winternitz,XMSS} are dominated by hashing: a single signature in the stateful multi-signature scheme proposed for Ethereum's quantum transition costs about $\bench{const.hashes_per_sig}$ hash invocations~\cite{DKKW25}, and aggregating many such signatures is exactly the batched-hashing workload \sysname targets.
Were Ethereum to instantiate these signatures with BLAKE3, \sysname's ten-core throughput of $\bench{bl.flock.mt.head}$k compressions per second could prove hashing corresponding to $\bench{app.eth.tps}$ transactions per second.
Vitalik Buterin recently estimated~\cite{vitalik2024} that proving $\bench{const.vitalik}$k hashes per second would suffice for a post-quantum transition.
\sysname, therefore, puts BLAKE3 and other standard hash functions as feasible targets for Ethereum's roadmap.
For Bitcoin, which is extremly conservative, and avoids novel cryptographic assumptions, a post-quantum transition will likely involve SHA256-based signatures. \sysname can prove hashing corresponding to $\bench{const.btc.sigs}$ signatures, the current upper limit of a Bitcoin block, in less than $\bench{app.btc.seconds}$ seconds on consumer hardware. We stress that these figures count only the hashing inside signature verification --- not transaction execution, state access, or any other logic.

\paragraph{Security.} \sysname targets $\bench{meta.security}$ bits of soundness (for details see \cref{sec:eval}). We rely only on proximity-gap results in the proven regime \cite{BCHKS25}; the remaining soundness assumption
is the security of SHA-256, used internally for both Merkle commitments and Fiat-Shamir.

\paragraph{Limitations.}
\iffull \sysname is a research prototype and we do not recommend using it in production systems yet. \fi
\iffull It \else \sysname \fi supports only batched Boolean-circuit computations expressed as R1CS. We give specific instantiations (hash-chaining and Merkle trees) for the IO circuit $G$, but do not support general circuits yet (in particular, supporting signature aggregation would require additional work). The implementation was heavily AI-assisted and was targeted at demonstrating optimal performance on a specific system (Apple ARM processors), rather than generality. We expect that similar optimizations carry over to other architectures.

Lastly, we remark that the system currently only offers succinctness and not zero-knowledge but we expect that zero-knowledge could be added at a relatively low cost (e.g., via techniques from the recent works \cite{CFW26,DHRR26}).

\subsection{Overview of techniques}
The starting point for \sysname is a Spartan-like \cite{Setty20} (cf. \cite[Chapter 8]{Thaler22}) PIOP for batch-R1CS, over the binary field $\Ftwo$. The batching approach is similar to that of \cite{BCGGRS19} (in the univariate setting) and \cite{TKPS22,HR22} (in the multilinear one). Compared to Plonkish~\cite{Plonk} arithmetization, R1CS has a significantly smaller witness as it does not commit to the output of addition gates. For soundness, we rely on binary extension fields. The use of such extension fields in proof-systems originates in early works \cite{BS08,FRI} leading into implemented systems \cite{BBHR18,Irr,hashcaster}.

Our concrete gains come from multiple ingredients, including multiple known optimizations combined with several novel ones that we introduce. We define a simple yet powerful zerocheck formulation for batch-R1CS that can be combined with a single small lincheck across the $A,B,C$ matrices. The lincheck amortizes the verifier's matrix evaluation cost over the batch dimension, and is extremely cheap. It takes less than $10\%$ of our total prover time. The batch R1CS optimization is based on \cite{BCGGRS19,TKPS22,HR22} and is quite simple, but already buys us significant mileage --- in particular making the zerocheck protocol the key remaining bottleneck, and the focus of most of our optimization efforts.

Focusing on zerocheck (which, in our context, is simply checking that the pointwise product between bit-vectors $a$ and $b$ is equal to $c$) we first and foremost utilize many known optimizations from the literature. These include the univariate skip \cite{Gruen24}, partially deterministic sumcheck challenges \cite{DT24}, deferred modular reduction and evaluation at infinity \cite{DBEMPT26}.

We push these optimizations further. First, building on Dao and Thaler's~\cite{DT24} partially deterministic zerocheck challenges, we identify a new set of challenge points that can process $\log(|\F|)$ of the individual zerocheck bit claims using roughly $2$ bit-shifts and XORs per input, and without relying on tower-fields, which are less efficient than direct extensions of $\Ftwo$.

Second, similarly to Binius64~\cite{Irr}, we utilize a degree 64 univariate skip~\cite{Gruen24}. We introduce a new cache-friendly lookup-based low-degree-extension kernel that is the main workhorse of the zerocheck protocol. 
Additionally, we show how the zerocheck claim over the vector $c$ can be \emph{elided}, removing roughly a third of the invocations of our most expensive procedure, namely, the batch univariate low degree extension for the univariate skip. This step requires us to evaluate our witness at two independent points. We further optimize the batch evaluation and the ring-switching techniques for this setting. 
See \cref{sec:impl} for details.

Our PCS is based on the Binius~\cite{biniusml} ring-switching technique combined with Ligerito~\cite{NA25} as the underlying dense PCS (the system is modular and one could equally instantiate the dense PCS with a different Boolean commitment scheme such as Blaze~\cite{Blaze} or Bolt~\cite{GNR26} or use Blaze as a Boolean-PCS and avoid the ring-switching altogether). \iffull We provide a self-contained overview of the ring-switching technique in \cref{appx:ring_switch}, replacing one of the steps with a more modular, and arguably simpler, approach. We also give an overview of Ligerito in \cref{appx:ligerito}, and extend its analysis to the list-decoding regime (following \cite{ArnonCFY25}).

\fi


\subsection{Related works}

\paragraph{Binius.} \sysname is most closely related to the Binius64 proof-system \cite{Irr} (which builds on \cite{binius,biniusml}). Like Binius, we work over binary fields and we build our $\F_2$ polynomial commitment using their packing and ring-switching technique~\cite{biniusml}, which reduces building a commitment to a small-field ($\F_2$) witness to building one over the large field $\F_{2^{128}}$. Unlike the Binius academic papers, which use tower fields, but similarly to the Binius64 implementation, we instantiate $\Feight$ as the AES field
and $\Fbig$ as GHASH, as they have much faster implementations. Some high-level differences between our approach and that of Binius64 are that we commit only to the inputs, outputs and multiplication gates, whereas Binius64 commits to all wire values. This leads to our trace being significantly smaller. Also, by focusing on the batch setting our lincheck is drastically cheaper than the more general one that they support. Lastly, our new techniques significantly improve zero-check, a key bottleneck for both systems but especially for \sysname.

\paragraph{Hashcaster.} Hashcaster~\cite{hashcaster} is a GKR-based system using binary fields that specifically targets batch proving of Keccak permutations. Extending it to SHA-256 or BLAKE3 seems difficult as these hash functions perform u32-addition operations that are hard to handle via GKR.\footnote{While logarithmic-depth linear-size circuits for addition are known \cite{BK82}, using them would introduces an impractical depth blowup.}

For proving Keccak, Hashcaster's performance is closest to ours ($\sim \hspace{-1mm}\bench{sp.kc.st.hc}\times$ slower on a single core, $\bench{sp.kc.mt.hc}\times$ slower multi-threaded). Additionally, and in contrast to \sysname, its proof size and verifier time grow linearly in the depth of the computation. The main advantage of the GKR approach is that the commitment is to a \emph{much} smaller witness. Interestingly, some of our key optimizations do not immediately translate to the GKR setting. In particular, our deterministic challenges rely on performing a single zerocheck-to-sumcheck conversion: because we move from a small $\F_2$ witness to an $\F_{2^{128}}$ sumcheck in one shot, the challenges for that step can be made partially deterministic. GKR-style systems such as Hashcaster instead run many rounds of sumcheck, with the challenges in each round derived from the previous one, so this optimization does not seem to carry over.

\paragraph{Prime field SNARKs.} A large body of work builds SNARKs over larger prime fields (and extensions thereof). Plonky3~\cite{Plonky3}, for example, combines a Plonk-style, univariate PIOP with FRI as its polynomial commitment. Arithmetizing Boolean computation --- such as the bit-level operations underlying standard hash functions --- is costly over these fields because of their high \emph{embedding overhead}: the gap between the number of bits a field element occupies and the number of bits of useful information it carries. While several of our techniques transfer to other fields, we focus on binary fields precisely because of their minimal embedding overhead. Overall, \sysname is around \bench{range.p3.lo}-\bench{range.p3.hi}x faster than Plonky3\cite{Plonky3}, a popular proof library which operates over a $31$-bit prime field.

\paragraph{Group based SNARKs.} Finally, there is a large family of group-based SNARKs, including Groth16~\cite{Groth16}, Plonk/HyperPlonk~\cite{Plonk,HyperPlonk23}, Spartan-Hyrax~\cite{Setty20}, and NeutronNova~\cite{NeutronNova}. These systems are not post-quantum secure and require arithmetizing the computation over an exponentially large prime field. While their proofs can be very small, their prover performance is significantly worse on batch hashing compared to \sysname. We compare against the \emph{multi-circuit} scheme from the vega-prover library~\cite{Spartan2}, an implementation of SplitNeutronNova\cite{Vega,NeutronNova}. The implementation does not contain ARM-specific optimiations and are designed for client-side proving, in particular the mobile driver's license application~\cite{Vega}, which contains a SHA-256 hash\footnote{The library is optimized for low ZK proving latency on signed messages,  not for raw throughput.\url{https://github.com/microsoft/vega-prover/blob/main/README.md}}. Nevertheless, the scheme is the fastest group-based prover for SHA-256 we found (vs. vega-sc, gnark, noir, snark-js). It is $\bench{sp.sha.st.nn.h}/\bench{sp.sha.mt.nn.h}\times$ (single/multi-threaded) slower than \sysname.

\subsection{Organization}
\Cref{sec:prelim} contains preliminaries. \Cref{sec:construction} presents a baseline Spartan-like proof-system for batch-R1CS over $\Ftwo$ that serves as our starting point. \Cref{sec:impl} is the technical core of the paper, developing the optimizations that make \sysname{} fast. \Cref{sec:eval} reports our experimental evaluation and comparison against prior systems.
\iffull
\cref{app:full-protocol} gives the full protocol specification, \cref{appx:ring_switch} a self-contained overview of the ring-switching technique. \cref{appx:ligerito} gives an overview of the Ligerito polynomial commitment scheme. Lastly, \cref{app:cauchy-shift} proves the Cauchy-shift identity.
\fi

\section{Preliminaries}
\label{sec:prelim}

\paragraph{Notation.} Throughout, $\F$ denotes a finite field. For a positive integer $n$, we write $[n] := \{1, \ldots, n\}$ and identify the Boolean hypercube $\bitset^n$ with the set of binary vectors of length $n$. Vectors are written in bold (e.g.\ $\mathbf{x}$) and indexed as $\mathbf{x} = (x_1, \ldots, x_n)$. We use $\log$ for the base-$2$ logarithm.

\subsection{Multilinear extensions}
\label{sec:prelim:mle}

Every function $f : \bitset^m \to \F$ has a unique multilinear extension $\hat{f} : \F^m \to \F$ that agrees with $f$ on $\bitset^m$, given explicitly by
\[
  \hat{f}(\mathbf{x}) \;=\; \sum_{\mathbf{b} \in \bitset^m} f(\mathbf{b})\cdot\eq(\mathbf{b}, \mathbf{x}),
\]
where the equality polynomial
\[
  \eq(\mathbf{b}, \mathbf{x}) \;:=\; \prod_{i=1}^{m} \bigl( b_i\, x_i + (1 - b_i)(1 - x_i) \bigr)
\]
is the unique multilinear polynomial that takes value $1$ on $\mathbf{b}$ and $0$ elsewhere on $\bitset^m$. The following equivalent representation will also be useful:

\begin{fact}\label{fact:eq-product}
  For $\mathbf{b} \in \bitset^m$ and $\mathbf{x} \in \F^m$ it holds that:
  \[
  \eq(\mathbf{b}, \mathbf{x}) = \prod_{i=1}^{m} \left( x_i^{b_i} \cdot (1-x_i)^{1-b_i} \right).
  \]
\end{fact}

\subsection{Rank-1 constraint systems (R1CS)}
\label{sec:prelim:r1cs}

A \emph{rank-1 constraint system} (R1CS) over $\F$ is specified by three matrices $A, B, C \in \F^{N \times M}$. A vector $\mathbf{z} \in \F^M$ \emph{satisfies} the system if
\[
  (A\mathbf{z}) \circ (B\mathbf{z}) \;=\; C\mathbf{z},
\]
where $\circ$ denotes the entrywise (Hadamard) product. Splitting $\mathbf{z} = (\mathbf{x}, \mathbf{w})$ into a public statement $\mathbf{x} \in \F^{\ell}$ and a private witness $\mathbf{w} \in \F^{M-\ell}$ yields the standard R1CS satisfiability relation: $\bigl((A,B,C), \mathbf{x}\bigr)$ is in the relation iff there exists $\mathbf{w}$ such that $\mathbf{z} = (\mathbf{x}, \mathbf{w})$ satisfies the above.

Throughout this paper we will be focusing on square R1CS matrices with $N = M = 2^m$ over the binary field $\Ftwo$.

\subsection{Sumcheck}
\label{sec:prelim:sumcheck}

The sumcheck protocol~\cite{LFKN92} is an interactive reduction for claims of the form
\[
  S \;=\; \sum_{\mathbf{b} \in \bitset^m} g(\mathbf{b}),
\]
where $g : \F^m \to \F$ is a polynomial of individual degree at most $d$ in each variable. The protocol runs for $m$ rounds; in each round the prover sends a univariate polynomial of degree $\le d$ and the verifier replies with a uniform challenge $r_i \in \F$. At the end the verifier is left with a single evaluation claim $g(\mathbf{r}) = v$ at a point $\mathbf{r} = (r_1, \ldots, r_m) \in \F^m$. The protocol has perfect completeness and soundness error at most $m d / |\F|$.

\begin{remark}[Batch sumcheck]
\label{rem:batch-sumcheck}
Multiple sumcheck claims $\{v_t = \sum_{\mathbf{b} \in \bitset^m} g_t(\mathbf{b})\}_{t=1}^T$ on the same hypercube can be folded into one: the verifier samples random $\alpha_1, \ldots, \alpha_T \in \F$, and the parties run sumcheck on $\sum_t \alpha_t v_t = \sum_{\mathbf{b}} \sum_t \alpha_t g_t(\mathbf{b})$. The combined sumcheck inherits the degree and round count; the soundness error increases additively by $1/|\F|$ due to the randomness of the $\alpha$-combination.
\end{remark}

\subsubsection{Zerocheck}
\label{sec:prelim:sumcheck:zero}
To check that a polynomial $f : \F^m \to \F$ vanishes on the Boolean hypercube, the verifier samples $\mathbf{r} \in \F^m$ and runs sumcheck on
\[
  \sum_{\mathbf{b} \in \bitset^m} \eq(\mathbf{r}, \mathbf{b}) \cdot f(\mathbf{b}) \;\stackrel{?}{=}\; 0.
\]
The summed expression equals the multilinear extension (in $\mathbf{r}$) of $f|_{\bitset^m}$, so it vanishes identically iff $f$ vanishes on $\bitset^m$; a non-vanishing $f$ is caught except with probability $m/|\F|$ over the choice of $\mathbf{r}$, plus the soundness error of the underlying sumcheck. We refer to this combined protocol as the \emph{zerocheck} on $f$.

\subsection{Multilinear polynomial commitments}
\label{sec:prelim:pcs}

A multilinear \emph{polynomial commitment scheme} (PCS) over $\F$ lets a prover commit to a multilinear polynomial $\hat{f} : \F^m \to \F$, producing a short digest $\mathsf{cm}$.
Later, the prover can generate a short proof $\pi$ certifying that $\hat{f}(\mathbf{r}) = v$ for a point $\mathbf{r} \in \F^m$ and value $v \in \F$.

\paragraph{Interactive oracle PCS.}
We describe the information-theoretic analogue of a PCS, known as an \emph{interactive oracle PCS} (IOPCS)~\cite{biniusml,BFRW25}.
A multilinear IOPCS consists of two phases:
\begin{itemize}
  \item \emph{Commitment phase:} The prover, given a function $f \colon \bitset^m \to \F$, interacts over multiple rounds with the verifier; the resulting transcript functions as a commitment to $f$.
  \item \emph{Opening phase:} Given a point $\mathbf{r} \in \F^m$ and value $v \in \F$ claimed to be $\hat{f}(\mathbf{r})$, the prover interacts over multiple rounds with the verifier; the resulting transcript functions as a proof. At the end, the verifier either accepts or rejects.
\end{itemize}
Crucially, we require that the verifier only reads a few bits from the commitment and proof strings, which may be long (e.g., as long as the description of $f$).
We additionally require the following properties:
\begin{itemize}
  \item \emph{Completeness:} Assuming the prover and verifier behave honestly, the verifier always accepts whenever $\hat{f}(\mathbf{r}) = v$.
  \item \emph{Binding:} Every commitment transcript $\mathsf{cm}$ is associated with a subset $S_{\mathsf{cm}}$ of multilinear polynomials $\F^{m} \to \F$ such that the following holds. For any (possibly dishonest) prover, $S_{\mathsf{cm}}$ contains at most one element, with all but negligible probability (over the prover and verifier's randomness).
  \item \emph{Soundness:} For any (possibly dishonest) prover, if there does not exist a $\hat{f} \in S_{\mathsf{cm}}$ such that $\hat{f}_{\mathsf{cm}}(\mathbf{r}) = v$, then the verifier rejects with all but negligible probability.
\end{itemize}
We remark that the opening phase (and soundness property) can be generalized to support multiple evaluation claims.

Throughout this work we refer to IOPCSs simply as PCSs.

\paragraph{Boolean PCS.} In this work we rely on \emph{Boolean} multilinear PCSs, which support committing to multilinear extensions of Boolean-valued functions $f:\bitset^{m} \to \bitset$. The verifier is guaranteed that any commitment is consistent with a function of this form.

\emph{Ring-switching}~\cite{biniusml} is a technique that enables constructing Boolean multilinear PCSs from (standard) multilinear PCSs.
In more detail, suppose that $\F \supseteq \F_2$ is a large binary field with extension degree $2^{k}$.
Then, ring-switching transforms any ``dense'' $(m-k)$-variable multilinear PCS over $\F$ into a Boolean $m$-variable multilinear PCS over $\F$.
In our constructions, we instantiate the dense PCS with Ligerito~\cite{NA25}.

\section{Baseline R1CS Protocol}
\label{sec:construction}

We start by describing a baseline succinct argument for R1CS over the binary field $\Ftwo$, deferring all optimizations to \cref{sec:impl}. The protocol is a variant of Spartan~\cite{Setty20} for the binary field setting (see also~\cite[Chapter~8]{Thaler22}). The construction is described as an interactive oracle proof~\cite{BCS16,RRR21}, which can be compiled into a SNARK via standard transformations~\cite{Kilian92,Micali00,BCS16}.

\paragraph{Setup.} Let $m \in \N$ and fix an R1CS instance $(A, B, C)$ in which the matrices are square matrices over $\F_2$ of dimension $2^m \times 2^m$. We view them as functions $A,B,C : \bitset^m \times \bitset^{m} \to \F_2$. Consider a satisfying assignment $\mathbf{z}$ of the form $\mathbf{z} = (\mathbf{x}, \mathbf{w}) \in \F_{2}^{2^m}$, where $\mathbf{x}$ is the public input (known to the verifier) and $\mathbf{w}$ is the private input (i.e., the witness). Let $\mathbf{a} := A\mathbf{z}$, $\mathbf{b} := B\mathbf{z}$, $\mathbf{c} := C\mathbf{z}$. We view $\mathbf{a}, \mathbf{b}, \mathbf{c}, \mathbf{z}$ as functions (from $\bitset^m$ to $\Ftwo$) and write $\hat{a}, \hat{b}, \hat{c}, \hat{z} : \F^m \to \F$ for their multilinear extensions over a large extension field $\F \supseteq \Ftwo$. Similarly, for each matrix $M \in \{A, B, C\}$ we denote its multilinear extension by $\hat{M} : \F^m \times \F^m \to \F$. The R1CS satisfiability condition is then equivalent to
\[
  \hat{a}(\mathbf{i}) \cdot \hat{b}(\mathbf{i}) - \hat{c}(\mathbf{i}) \;=\; 0 \qquad \forall\, \mathbf{i} \in \bitset^m.
\]

\paragraph{Protocol.} The baseline protocol proceeds as follows:

\begin{enumerate}
\item \textbf{Commit to the trace.} The prover commits to $\hat{z}$ using a Boolean multilinear PCS (see \cref{sec:prelim:pcs}), producing a commitment transcript $\mathsf{cm}$.

\item \textbf{Zerocheck on $\hat{a} \cdot \hat{b} - \hat{c}$.} The verifier samples a random challenge $\mathbf{r} \in \F^m$, and the parties run sumcheck on
\[
  \sum_{\mathbf{i} \in \bitset^m} \eq(\mathbf{r}, \mathbf{i}) \cdot \bigl(\hat{a}(\mathbf{i}) \cdot \hat{b}(\mathbf{i}) - \hat{c}(\mathbf{i})\bigr) \;\stackrel{?}{=}\; 0.
\]
After $m$ rounds the verifier holds a random point $\mathbf{r}_y \in \F^m$ and claimed values
\[
  v_a = \hat{a}(\mathbf{r}_y), \qquad v_b = \hat{b}(\mathbf{r}_y), \qquad v_c = \hat{c}(\mathbf{r}_y),
\]
and checks locally that $\eq(\mathbf{r}, \mathbf{r}_y) \cdot (v_a \cdot v_b - v_c)$ matches the final sumcheck value.

\item \textbf{Lincheck: reduce the $\hat{a}, \hat{b}, \hat{c}$ claims to a single claim on $\hat{z}$.} We describe the reduction for $\hat{a}$ first. Since $\mathbf{a} = A\mathbf{z}$, the MLE evaluation of $\hat{a}$ at $\mathbf{r}_y$ can be expressed as:
\[
  \hat{a}(\mathbf{r}_y) \;=\; \sum_{\mathbf{j} \in \bitset^m} \hat{A}(\mathbf{r}_y, \mathbf{j}) \cdot \hat{z}(\mathbf{j}).
\]
The parties can therefore check the claim $v_a = \hat{a}(\mathbf{r}_y)$ via a sumcheck, resulting in a verifier evaluation of $\hat{A}(\mathbf{r}_y, \mathbf{r}_x) \cdot \hat{z}(\mathbf{r}_x)$ at a random $\mathbf{r}_x \in \F^m$. For the moment let us assume that the verifier computes $\hat{A}(\mathbf{r}_y, \mathbf{r}_x)$ on its own and so we are left only with a claim on $\hat{z}(\mathbf{r}_x)$.\footnote{Since $A$ is a fixed matrix, this computation can be done independently by the verifier. However, for a general matrix $A$ it can be costly. In the actual protocol we will leverage the fact that the matrix is structured to reduce the verifier's cost (see \cref{sec:impl:batch}).} We refer to such a reduction---from an MLE claim on $\hat{a}$ to an MLE claim on $\hat{z}$ via the matrix MLE $\hat{A}$---as a \emph{lincheck} on $A$.

The same applies with $B, C$ in place of $A$. Batching the three linchecks via \cref{rem:batch-sumcheck}, the three lincheck claims reduce to a single claim $v_z = \hat{z}(\mathbf{r}_x)$ at a random point $\mathbf{r}_x \in \F^m$ as well as evaluation claims on $\hat{A},\hat{B},\hat{C}$.

\item \textbf{Consistency with public input.}\footnote{This step will be replaced with input/output constraints (\Cref{sec:impl:io}) when we move to Batch R1CS in the next section.} For simplicity, assume that the public input length is $2^{k}$ and $\mathbf{z}$ is arranged so that $\mathbf{z}(0^{m-k},\cdot)$ should be equal to $\mathbf{x}$. The verifier tests consistency via an MLE claim $v' = \hat{z}(\mathbf{r}')$, where $\mathbf{r}' = (0^{m-k}, \mathbf{r})$ for randomly sampled $\mathbf{r} \in \F^{k}$ and $v' = \hat{x}(\mathbf{r})$ is computed directly by the verifier.

\item \textbf{Open the trace.} The prover and verifier run the PCS opening phase on $\mathsf{cm}$ with evaluations $(\mathbf{r}_x, v_z)$ and $(\mathbf{r'}, v')$. The verifier accepts if and only if PCS verifier accepts and every check above passes.
\end{enumerate}

\paragraph{Completeness.}
Let $\mathbf{x}$ be a public input and suppose the (honest) prover holds a witness $\mathbf{w}$ with $\mathbf{z} = (\mathbf{x}, \mathbf{w})$ being a satisfying assignment.
Then $\eq(\mathbf{r}, \mathbf{i}) \cdot (\hat{a}(\mathbf{i}) \cdot \hat{b}(\mathbf{i}) - \hat{c}(\mathbf{i}))$ vanishes on $\bitset^m$, and, by completeness of zerocheck and lincheck, it holds that $v_{z} = \hat{z}(\mathbf{r}_{x})$.
Since $\mathbf{z}(0^{m-k}, \cdot) \equiv \mathbf{x}$, it also holds that $v' = \hat{z}(\mathbf{r}')$.
Finally, by completeness of the PCS, the verifier always accepts.

\paragraph{Soundness.}
Let $\mathbf{x}$ be a public input and suppose there does not exist a witness $\mathbf{w}$ with $(\mathbf{x}, \mathbf{w})$ being a satisfying assignment.
With all but the PCS's binding error probability, the commitment transcript $\mathsf{cm}$ is associated with a subset $S_{\mathsf{cm}}$ containing at most one multilinear polynomial.
If $S_{\mathsf{cm}}$ is empty, then the verifier rejects with all but the PCS's soundness error probability.
Thus, we focus on the scenario where $S_{\mathsf{cm}}$ contains only the multilinear extension of some function $\mathbf{z}_{\mathsf{cm}} \colon \bitset^m \to \F$.
There are two cases:
\begin{itemize}
  \item If $\mathbf{z}_{\mathsf{cm}}$ is a satisfying assignment, then it must be inconsistent with $\mathbf{x}$, i.e., $\hat{z}_{\mathsf{cm}}(\cdot, 0^{m-k}) \not\equiv \mathbf{x}$. By Schwartz-Zippel, with all but $k/|\F|$ probability, it holds that $\hat{z}_{\mathsf{cm}}(\mathbf{r}') \neq v'$. Assuming this, the verifier rejects with all but the PCS's soundness error probability.
  \item If $\mathbf{z}_{\mathsf{cm}}$ is not a satisfying assignment, then there exists $\mathbf{i} \in \bitset^{m}$ such that $\hat{a}_{\mathsf{cm}}(\mathbf{i}) \cdot \hat{b}_{\mathsf{cm}}(\mathbf{i}) - \hat{c}_{\mathsf{cm}}(\mathbf{i}) \neq 0$, where $\hat{a}_{\mathsf{cm}}, \hat{b}_{\mathsf{cm}}, \hat{c}_{\mathsf{cm}}$ are the multilinear extensions of $A \mathbf{z}_{\mathsf{cm}}, B \mathbf{z}_{\mathsf{cm}}, C \mathbf{z}_{\mathsf{cm}}$.
    The zerocheck and lincheck protocols are ran with respect to $\mathbf{z}_{\mathsf{cm}}$; with all probability $O(m/|\F|)$, it holds that $\hat{z}_{\mathsf{cm}}(\mathbf{r}_{x}) \neq v_{z}$ (or one of the verifier's other checks already fails).
    Assuming this, the verifier rejects with all but the PCS's soundness error probability.
\end{itemize}
We conclude that the overall soundness error is at most $O(m/|\F|)$, plus the PCS's binding and soundness errors.

\begin{remark}
  A straightforward implementation of the baseline protocol over $\Ftwo$ is quite inefficient. In particular, when used over $\F_2$, the zerocheck protocol introduces an exorbitant $O(2^m \cdot \log(|\F|))$ prover cost. Indeed, this is why the Spartan protocol is designed for R1CS over large fields. In the next section we discuss how to efficiently prove R1CS over $\F_{2}$.
\end{remark}

\section{Flock}\label{sec:impl}

We describe the different optimizations that we utilize in \sysname to make the baseline protocol extremely efficient and practical. These optimizations build on, and improve, a long line of prior research in the field.
For completeness, we describe the full protocol in Appendix~\ref{app:full-protocol}.

Since some of the optimizations are related to each other, we do not list them in order of significance but rather based on a topological sort of the dependency graph.

\subsection{Batch R1CS}\label{sec:impl:batch}

The motivating workloads for \sysname are \emph{batch computations}: $K$ independent invocations of the same primitive --- e.g., multiple Keccak permutations or SHA2 compressions. Each individual computation can be described by the same R1CS instance over $\Ftwo$. Rather than proving the $K$ instances separately, we stack them into a single R1CS with dimension that is $K$ times larger. Utilizing a technique from \cite{BCGGRS19} (extended to the multilinear setting \cite{TKPS22,HR22}), such a batch-R1CS can be handled \emph{much} more efficiently than a generic instance.

Concretely, let the per-invocation (\emph{base}) matrices be $A_0, B_0, C_0 : \bitset^{m_0} \times \bitset^{m_0} \to \F_2$ and let $K = 2^k$ denote the number of instances. These matrices correspond to verification of a single copy. The batched witness is
\[
  \mathbf{z} : \bitset^k \times \bitset^{m_0} \to \F_2
\]
where we denote by $\mathbf{z}_i = \mathbf{z}(i,\cdot)$. The batched matrices are the block-diagonal matrices
\[
  A \;=\; I_K \otimes A_0, \quad B \;=\; I_K \otimes B_0, \quad C \;=\; I_K \otimes C_0,
\]
where $I_K$ denotes the $K \times K$ identity matrix and $\otimes$ is the tensor product. To see this observe that $(I_K \otimes A_0) \cdot z$ simply applies the matrix $A_0$ to each component of $z$ separately.

These are the matrices fed to the baseline protocol of \cref{sec:construction}, with the role of $m$ played by $m := k + m_0$. Satisfaction of $(A\mathbf{z}) \circ (B\mathbf{z}) = C\mathbf{z}$ is equivalent to simultaneous satisfaction of all $K$ underlying instances. We refer to $A\mathbf{z}$, $B\mathbf{z}$, and $C\mathbf{z}$ as the \emph{$A$-side}, \emph{$B$-side}, and \emph{$C$-side} of the constraint, respectively.

\paragraph{Lincheck via block structure.} Decompose each hypercube index $\mathbf{i} \in \bitset^m$ as $\mathbf{i} = (\mathbf{i}_{\mathrm{out}}, \mathbf{i}_{\mathrm{in}})$ with $\mathbf{i}_{\mathrm{out}} \in \bitset^k$ the batch index and $\mathbf{i}_{\mathrm{in}} \in \bitset^{m_0}$ the within-instance index. The MLE of a block-diagonal matrix factors as
\[
  \hat{A}\bigl((\mathbf{r}_{\mathrm{out}}, \mathbf{r}_{\mathrm{in}}),\, (\mathbf{j}_{\mathrm{out}}, \mathbf{j}_{\mathrm{in}})\bigr) \;=\; \eq(\mathbf{r}_{\mathrm{out}}, \mathbf{j}_{\mathrm{out}}) \cdot \hat{A}_0(\mathbf{r}_{\mathrm{in}}, \mathbf{j}_{\mathrm{in}}),
\]
and similarly for $\hat{B}, \hat{C}$. Substituting into the lincheck-on-$A$ identity from \cref{sec:construction} and pulling the $\mathbf{j}_{\mathrm{out}}$-only factor out of the inner sum,
\begin{align*}
  \hat{a}(\mathbf{r}_y) &= \sum_{\mathbf{j}_{\mathrm{out}}} \eq(\mathbf{r}_{y,\mathrm{out}}, \mathbf{j}_{\mathrm{out}})
  \cdot \sum_{\mathbf{j}_{\mathrm{in}}} \hat{A}_0(\mathbf{r}_{y,\mathrm{in}}, \mathbf{j}_{\mathrm{in}}) \cdot \hat{z}(\mathbf{j}_{\mathrm{out}}, \mathbf{j}_{\mathrm{in}}) \\
  &= \sum_{\mathbf{j}_{\mathrm{in}}} \hat{A}_0(\mathbf{r}_{y,\mathrm{in}}, \mathbf{j}_{\mathrm{in}}) \cdot \hat{z}(\mathbf{r}_{y,\mathrm{out}}, \mathbf{j}_{\mathrm{in}}),
\end{align*}
a sumcheck over only $m_0$ variables (rather than $m = k + m_0$). Two key consequences:
\begin{itemize}
  \item The verifier evaluates only the base MLEs $\hat{A}_0, \hat{B}_0, \hat{C}_0$ at a single point $(\mathbf{r}_{y,\mathrm{in}}, \mathbf{r}_{x,\mathrm{in}}) \in \F^{m_0} \times \F^{m_0}$; the outer $k$ rounds contribute only a cheap extra $\eq$-factor evaluation.

  In particular, the matrices are sufficiently small that the verifier can compute their multilinear extension at the desired point by itself.

  \item Even more significantly, the prover only needs to ``fold'' $z$ into $\hat{z}(\mathbf{r}_{y,\mathrm{out}}, \cdot )$ and the remaining work is independent of $k$; only the inner $m_0$ rounds operate on the small base matrices $A_0, B_0, C_0$ (rather than the $K$-times-larger matrices $A, B, C$).
\end{itemize}

The latter point is crucial and drastically reduces the lincheck cost in \sysname. In particular, we observe that when using this batching technique the dominant part in the protocol by far becomes zerocheck, and so most of the optimizations described below attempt to extensively optimize that part of the protocol.

\subsection{Univariate skip}\label{sec:impl:uskip}

\paragraph{Motivation.} The zerocheck of \cref{sec:construction} runs sumcheck on a summand $\eq(\mathbf{r}, \mathbf{i}) \cdot (\hat{a}(\mathbf{i}) \hat{b}(\mathbf{i}) - \hat{c}(\mathbf{i}))$ whose constraint values $\hat{a}, \hat{b}, \hat{c}$ are bits. Recall that after the initial round of sumcheck, the verifier samples a challenge $r_1 \in \F$ from a large extension field (e.g.\ $|\F| = 2^{128}$), and the residual claim is a sumcheck over $\bitset^{m-1}$ on the partial evaluations $\hat{a}(r_1, \cdot), \hat{b}(r_1, \cdot), \hat{c}(r_1, \cdot)$. The hypercube halves, but each cell is now an element of $\F$ rather than a bit: a size-$n = 2^m$ \emph{bit}-valued claim has been replaced by a size-$n/2$ \emph{$\F$}-valued claim --- a $\log(|\F|)/2$ blowup of the working state. This is a massively unaffordable cost.

\paragraph{The univariate skip.} We utilize a popular technique due to Gruen~\cite{Gruen24}, called \emph{the univariate skip}.\footnote{A related technique, also geared at reducing the cost of sumcheck over $\F_2$, was proposed earlier by Holmgren and Rothblum \cite[Lemma 6.4]{HR22}. Their technique is asymptotically faster than the univariate skip, but seems concretely worse as it requires sending an additional oracle.}

Fix a parameter $k_{\mathrm{s}} \in [m]$ (concretely $k_{\mathrm{s}} = 6$ in our implementation) and view each of $a, b, c : \bitset^m \to \F_2$ as a matrix of shape $2^{m - k_{\mathrm{s}}} \times 2^{k_{\mathrm{s}}}$, indexed by a row coordinate $\mathbf{x} \in \bitset^{m - k_{\mathrm{s}}}$ (the variables \emph{kept} for the residual sumcheck) and a column coordinate $\mathbf{y} \in \bitset^{k_{\mathrm{s}}}$ (the variables \emph{skipped} in the first round). Fix a $k_{\mathrm{s}}$-dimensional $\F_2$-subspace $S \subset \F$ of size $2^{k_{\mathrm{s}}}$ and identify the $2^{k_{\mathrm{s}}}$ columns of the matrix $a$ with the points of $S$. We view each row $\mathbf{x}$ of $a$ as a table of evaluations of a unique univariate polynomial $\alpha_\mathbf{x} : \F \to \F_2$ of degree $< 2^{k_{\mathrm{s}}}$ at the points of $S$. We refer to $\alpha_\mathbf{x}$ as the \emph{low-degree extension} (LDE) of the row; analogously $\beta_\mathbf{x}, \gamma_\mathbf{x}$ are the LDEs of the corresponding rows of $b$ and $c$.

Instead of running $k_{\mathrm{s}}$ vanilla sumcheck rounds --- each consuming one large-field challenge --- we collapse them into a single round. The verifier samples a zerocheck challenge $\mathbf{r} \in \F^{m - k_{\mathrm{s}}}$ over the kept variables only, and the prover computes and sends the univariate
\begin{align*}
  P(\lambda) \;:=\; \sum_{\mathbf{x} \in \bitset^{m - k_{\mathrm{s}}}} \eq(\mathbf{r}, \mathbf{x}) \cdot \bigl( \alpha_\mathbf{x}(\lambda) \cdot \beta_\mathbf{x}(\lambda) - \gamma_\mathbf{x}(\lambda) \bigr)
\end{align*}
of degree $\le 2 \cdot (2^{k_{\mathrm{s}}} - 1)$, transmitted as evaluations at the points of a disjoint output coset $\Lambda := S + \delta \subset \F_{2^{k_{\mathrm{s}}}}$, for fixed $\delta \in \F \backslash S$. (It suffices to send evaluations on $S+\delta$ since $P$ is known to be identically $0$ on $S$.)

The verifier samples a single challenge $\lambda \in \F$ and the residual sumcheck of $m - k_{\mathrm{s}}$ rounds proceeds over $\F$ on the partially evaluated $\hat{a}(\lambda, \cdot), \hat{b}(\lambda, \cdot), \hat{c}(\lambda, \cdot)$.

The win is that the prover crosses the $\F_2 \to \F$ boundary exactly once. The $k_{\mathrm{s}}$ rounds of bit-level $\F_2$ structure are preserved end-to-end via the row-as-polynomial view, and the field upgrade is paid in a single round message rather than in $k_{\mathrm{s}}$ separate witness blowups.

\paragraph{LDE via lookup tables.} In the univariate skip round, the prover's hot loop computes, for each row $\mathbf{x}$, the LDE evaluations $\alpha_\mathbf{x}|_\Lambda, \beta_\mathbf{x}|_\Lambda, \gamma_\mathbf{x}|_\Lambda$ on the output coset. Since we extend $0/1$ valued functions, by taking $S \cup \Lambda$ to be contained in a subfield, we can ensure that all of the computed values are in the subfield.

Nevertheless, the question is how to compute this massive number of small LDEs. One option is to just use the NTT for every row; or, more precisely, an inverse NTT to recover the coefficients of each polynomial from its evaluations on $S$, and then a forward NTT to extend to $\Lambda$. This turns out to be highly inefficient in practice due to two reasons. First, the $O(|S| \log |S|)$ asymptotic cost of the NTT does not work well on the very small input sizes that we operate on. Second, the NTT does not exploit the fact that each row is bit-valued: it pays the full large-field arithmetic at every butterfly round.\footnote{One could potentially get some small savings by leveraging the small values in the first few butterfly rounds.}

Following \cite{Irr}, we can instead precompute the entire LDE map as a lookup table. Fix a chunk size $\tau \in \N$ in bits, such that $\tau$ divides $|S|$ and $\log(|\F|)$ (e.g., a good number to bear in mind is $\tau=8$, corresponding to a single byte). Taking $S \cup \Lambda$ to be contained in the subfield\footnote{Concretely, if $\tau=8$, we can use the AES field $\F_{2^8}$ \cite{Irr}. The usage of the AES field as a subfield, in a somewhat similar manner was first proposed by \cite{BBHR18}.} $\F_{2^\tau} \subseteq \F$, we can view the composite map on a single column as,
\[
  M \;:=\; \mathrm{NTT}_\Lambda \circ \mathrm{iNTT}_S \;:\; \F_{2^\tau}^{|S|} \;\to\; \F_{2^\tau}^{|\Lambda|},
\]
is $\F_2$-linear, so packing the row's bits into $|S|/\tau$ chunks $a_0(\mathbf{x}), \ldots, a_{|S|/\tau - 1}(\mathbf{x}) \in [0, 2^\tau)$ and tabulating $M$'s action per chunk position gives
\[
  \alpha_\mathbf{x}|_\Lambda \;=\; \bigoplus_b M_b\bigl[\,a_b(\mathbf{x})\,\bigr],
\]
where $M_b[v] \in \F_{2^\tau}^{|\Lambda|}$ encodes the action of $M$ on a chunk of value $v$ at position $b$. Each LDE call is then a fixed number of chunk lookups and XORs, sidestepping the NTT entirely.

\paragraph{New optimization: compressed lookup table.} Notice that the above approach does not leverage any structure of the matrix $M$. While we do not know how leverage the full NTT structure, we observe that $M$ still has useful structure that we can exploit, which reduces the size of the lookup table used by an $|S|/\tau$ factor. This is significant as we observe that the compressed table can be resident in the L1 cache, thereby giving a significant speedup in the hot loop.

Specifically, $M$ satisfies the following \emph{translation-invariance} property: shifting the column index by $\tau b$, for any $b$, is equivalent to XOR-shifting the row index by $\tau b$. Formally, for every $i \in [|\Lambda|]$, $b \in \{0, 1, \ldots, |S|/\tau - 1\}$, and $j \in \{0, 1, \ldots, \tau - 1\}$,
\[
  M\bigl[i,\, (\tau b) + j\bigr] \;=\; M\bigl[\,i \oplus (\tau b),\, j\,\bigr],
\]
where the column index arithmetic $(\tau b) + j$ is over the integers (in particular $\tau b$ stands for integer multiplication) whereas the row-index $\oplus$ is a bitwise XOR. The proof of this identity is deferred to Appendix~\ref{app:cauchy-shift}.

Define the \emph{single} base table $T$ where for $v \in \bitset^\tau$:
\[
  T[v] \;:=\; M \begin{pmatrix} v \\ \mathbf{0} \end{pmatrix} \;=\; \bigoplus_{j : v_j = 1} M[\cdot,\, j] \;\in\; \F_{2^\tau}^{|\Lambda|},
\]
where $\begin{pmatrix} v \\ \mathbf{0} \end{pmatrix}$ has $v$ in its first $\tau$ entries and zeros in the remaining $|S| - \tau$ entries. The translation-invariance of $M$ implies $M_b[v][i] = T[v]\bigl[i \oplus (\tau b)\bigr]$ for any chunk position $b$. Letting $r_b := T\bigl[a_b(\mathbf{x})\bigr] \in \F_{2^\tau}^{|\Lambda|}$ be the $b$-th chunk's row of the base table, the compressed-lookup form of the LDE is
\[
  \alpha_\mathbf{x}|_\Lambda [i] \;=\; \bigoplus_{b=0}^{|S|/\tau - 1} r_b\bigl[\,i \oplus (\tau b)\,\bigr],
\]
i.e.\ a single base-table access per chunk: instead of keeping $|S|/\tau$ position-specific tables, we read from the same table $T$ for every chunk and just XOR $\tau b$ into the index $i$ to account for the chunk position, which costs one additional XOR per output entry.

Concretely, taking $\tau = 8$ (so each chunk is a byte) and $|S| = |\Lambda| = 64$, the per-chunk-table version occupies $(|S|/\tau) \cdot 2^\tau \cdot |\Lambda| = 8 \cdot 256 \cdot 64 = 128\,\mathrm{KB}$, which may spill out of the L1 cache\footnote{On our benchmark M4 processors the L1 data cache is exactly 128KB.}; the compressed table is just $2^\tau \cdot |\Lambda| = 256 \cdot 64 = 16\,\mathrm{KB}$.

The two approaches do roughly the same arithmetic work per LDE call: $|S|/\tau$ chunk lookups returning length-$|\Lambda|$ vectors, accumulated by $|S|/\tau$ vector XORs of length $|\Lambda|$ --- i.e.\ $|S|/\tau \cdot |\Lambda|$ byte loads and $|S|/\tau \cdot |\Lambda|$ byte XORs in both. The compressed lookup adds only a constant cost per chunk to compute the index offset $\tau b$ (a single XOR amortized over the $|\Lambda|$ output entries), which is dwarfed by the main loop. The win is therefore not in operation count but in cache locality: one memory-hot base table reused across all chunk positions instead of $|S|/\tau$ separately materialized tables that compete for cache.

\subsection{Friendly challenges}
\label{sec:impl:friendly}

The zerocheck protocol (as described in \cref{sec:prelim:sumcheck:zero}) weighs the constraint polynomial by an eq-factor $\eq(\mathbf{r}, \cdot )$ for a challenge vector $\mathbf{r} \in \F^m$ that the verifier samples from the large extension field $\F$. Each coordinate $r_i$ shows up as a multiplier inside the prover's per-round arithmetic.

Motivated by the fact that multiplication by a generic $\F$ element is expensive, Dao and Thaler~\cite{DT24} observed that a small number of coordinates of $\mathbf{r}$ can be \emph{pinned} to fixed elements of $\F$.

Specifically, Dao and Thaler work over a tower field $\F$ and show that one can pin coordinates of $\mathbf{r}$ to the $7$ tower-level generators $x_1, \dots, x_7$, each of which lives in a (progressively larger) subfield of $\F$.

While this avoids some of the expensive multiplications, in general the tower representation seems significantly less efficient than other representations. For this reason Binius64~\cite{Irr} uses a non-tower representation (specifically they use the GHASH representation) and leverages the \cite{DT24} trick differently. Specifically they pin $3$ of the coordinates of $\mathbf{r}$ to fixed elements of the subfield $\F_{2^8} \subset \F$, multiplication by which is significantly cheaper than by a generic element of the large field $\F$. Specifically, viewing elements of $\F_{2^8}$ as univariates over $\F_2$, Binius chooses the three fixed elements $r_1=x$, $r_2=x^2$, and $r_3=x^4$.

\paragraph{Soundness.} We briefly recall the soundness argument from \cite{DT24}. Replacing seven coordinates of $\mathbf{r}$ with hardcoded constants restricts the zerocheck verifier's challenge from a uniformly random point in $\F^m$ to a uniformly random point in an $(m-7)$-dimensional affine subspace. Soundness still goes through, provided the seven pinned values are \emph{$\Ftwo$-linearly independent}. To see why, recall that the zerocheck reduces a constraint of the form ``$p(\mathbf{x}) = 0$ for all $\mathbf{x} \in \bitset^m$'' to a single evaluation claim at $\mathbf{r}$. If however $p$ is known to take $0/1$ values (which is the case for us), then $p$ vanishes on $\bitset^m$ if and only if $\sum_{\mathbf{b}\in\bitset^7}\alpha_\mathbf{b}\cdot p(\mathbf{b},\mathbf{x}) = 0$ for every $\mathbf{x}\in\bitset^{m-7}$, for every fixed sequence $\alpha_{\mathbf{b}}$ that are $\Ftwo$-linearly independent.

Thus, one can set $\mathbf{r} = (r_1,\dots,r_7) \in \F^7$ to be any fixed vector as long as the corresponding $(\eq(\mathbf{r},\mathbf{b}))_{\mathbf{b} \in \bitset^7}$ are linearly independent.

\paragraph{Our choice: a geometric progression.} We pick a different set of constant points. Let $d$ denote the extension degree of $\F$ (i.e., $\F$ is isomorphic to $\F_{2^d}$). For $i \in [d]$, we set
\[
  r_i \;:=\; \frac{x^{2^{i-1}}}{1 + x^{2^{i-1}}},
\]
where both the numerator and denominator refer to the polynomial representation of elements in $\F$ (and the division is over the field). The reason we do so is that $r_i/(1+r_i) = x^{2^{i-1}}$. Using this fact, we observe that for $\mathbf{r} = (r_1,\dots,r_d)$ and any $\mathbf{b} = (b_1,\dots,b_d) \in \bitset^d$, applying \cref{fact:eq-product}:
\begin{align*}
  \eq(\mathbf{r},\mathbf{b})
  &= \prod_{i \in [d]} (1 + r_i)^{1 - b_i} \cdot r_i^{b_i} \\
  &= C \cdot \prod_{i\in [d]} \bigl(x^{2^{i-1}}\bigr)^{b_i}
  \;=\; C \cdot x^{\mathrm{int}(\mathbf{b})},
\end{align*}
where $C = \prod_{i \in [d]} (1+r_i)$ is a fixed constant, and $\mathrm{int}(\mathbf{b}) := \sum_{i\in [d]} 2^{i-1}b_i$ is the integer corresponding to $\mathbf{b}$. Thus, the $2^d$ weights form the geometric progression:
\[
\{C,\ Cx,\ Cx^2,\dots,\ Cx^{2^d-1}\}.
\]

The benefit is structural: multiplication by $x^k$ in the field is just a $k$-bit shift followed by a modular reduction. To exploit this, the prover splits the sum into outer and inner parts:
\[
  \sum_{\mathbf{b} \in \bitset^m } \eq(\mathbf{r},\mathbf{b}) \cdot f(\mathbf{b}) \;=\; \sum_{\mathbf{b}_{\mathrm{out}} \in \bitset^{m-d}} \eq(\mathbf{r}_{\mathrm{out}},\mathbf{b}_{\mathrm{out}}) \cdot g(\mathbf{b}_{\mathrm{out}}),
\]
where the inner factor is
\[
  g(\mathbf{b}_{\mathrm{out}}) \;:=\; \sum_{\mathbf{b}_{\mathrm{in}} \in \bitset^{d}} \eq(\mathbf{r}_{\mathrm{in}},\mathbf{b}_{\mathrm{in}}) \cdot f(\mathbf{b}_{\mathrm{out}},\mathbf{b}_{\mathrm{in}}).
\]
Using our fixed choice of $\mathbf{r}_{\mathrm{in}}$, in computing $g$, the $2^d$ scalar multiplications by the $\eq(\mathbf{r}_{\mathrm{in}},\cdot)$ weights are replaced by $2^d$ bit shifts, followed by a single modular reduction at the end of each $2^d$-size chunk. The fixed scalar $C$ commutes through the outer sum and is absorbed once at protocol startup --- with zero cost in the hot loop.

\paragraph{Two-level implementation over $\Feight$ and $\Fbig$.} The discussion above takes place in a single extension field $\F$ for simplicity. In our implementation, following the heavily optimized implementation of \cite{Irr}, we use the field $\F = \Fbig$ in GHASH form $\Fbig = \Ftwo[\gamma] / (\gamma^{128} + \gamma^7 + \gamma^2 + \gamma + 1)$, and we exploit the subfield containment $\Feight \subset \Fbig$ to apply the trick \emph{twice}. Three pinned coordinates are placed in $\Feight$ using its standard AES representation $\Feight = \Ftwo[\alpha] / (\alpha^8 + \alpha^4 + \alpha^3 + \alpha + 1)$ with generator $\alpha$, and four more are placed in $\Fbig$ using $\gamma$. Thereby we derive a total of $2^3 \cdot 2^4 = 128$ pinned-eq weights of the form $C \cdot \alpha^k \cdot \gamma^j$ for $k \in [0,8)$ and $j \in [0,16)$, which replace $7$ coordinates of $\mathbf{r}$ in total.

The two levels are handled differently in the hot loop. The $\Feight$ level uses the shift-then-reduce of the previous paragraph, producing one $\Feight$ scalar $y_j$ per medium index $j \in [0, 16)$. The $\Fbig$ level, in contrast, is absorbed entirely into a precomputed lookup table $T[j][v] := \gamma^j \cdot \varphi_8(v) \in \Fbig$ for $j \in [0, 16)$ and $v \in [0, 256)$, where $\varphi_8 \colon \Feight \hookrightarrow \Fbig$ is the subfield embedding. Built once at protocol startup, this $16 \times 256$ table (a one-time $64$~KB precomputation) fuses the $\Feight$-to-$\Fbig$ conversion with the $16$-fold $\gamma$-progression. 
The constants commute through the outer sum and are absorbed once at protocol startup.

Importantly, we have verified that our fixed choice of constants does indeed generate linearly independent vectors over $\Ftwo$.

\subsection{Skipping $c$}
\label{sec:impl:c}

Recall that the goal in our use of the zerocheck protocol is to reduce checking that $\hat{a} \cdot \hat{b} - \hat{c}$ is identically zero on $\bitset^m$ into individual multilinear evaluation claims on $\hat{a}$, $\hat{b}$ and $\hat{c}$. In this section we describe an optimization that reduces the cost of processing $c$ in this protocol.

\paragraph{Skipping $c$: First Attempt.} A trivial observation is that we can rewrite the zerocheck expression
\[
  \sum_{\mathbf{i} \in \bitset^m} \eq(\mathbf{r}, \mathbf{i}) \cdot \bigl(\hat{a}(\mathbf{i}) \cdot \hat{b}(\mathbf{i}) - \hat{c}(\mathbf{i})\bigr) \;\stackrel{?}{=}\; 0
\]
as
\begin{align}\label{eq:skip_c}
  \sum_{\mathbf{i} \in \bitset^m} \eq(\mathbf{r}, \mathbf{i}) \cdot \hat{a}(\mathbf{i}) \cdot \hat{b}(\mathbf{i}) \stackrel{?}{=}   \sum_{\mathbf{i} \in \bitset^m} \eq(\mathbf{r}, \mathbf{i}) \cdot \hat{c}(\mathbf{i}),
\end{align}
and that the RHS is simply equal to $\hat{c}(\mathbf{r})$.

Thus, the prover computes $v=\hat{c}(\mathbf{r})$, sends it to the verifier and they run a sumcheck only on the LHS, avoiding any further processing of $c$. (Indeed, the claim $v=\hat{c}(\mathbf{r})$ is part of the output of the protocol.) At first glance this seems to reduce the overall cost of zerocheck by nearly $33\%$.

While this idea has merit, it also has a downside. Usually in the zerocheck protocol, when sending the first round sumcheck polynomial, the verifier knows a priori that its value on $0$ and $1$ have to be $0$ and so the prover can compute one less evaluation. When employing the univariate skip (see \cref{sec:impl:uskip}), with parameter $k_{\mathrm{s}}$, the zerocheck optimization means that we only need to report roughly $2^{k_{\mathrm{s}}}$ evaluations rather than twice that. Concretely, the original zerocheck round-1 message $P(\lambda)$ has degree $2(2^{k_{\mathrm{s}}} - 1)$ and vanishes on $S$, so the prover transmits only its $2^{k_{\mathrm{s}}}$ evaluations on the disjoint coset $\Lambda$. Once $c$ is split off as in \cref{eq:skip_c}, the LHS sumcheck is no longer a zerocheck (its target value is $\hat{c}(\mathbf{r})$ rather than $0$), so this saving is lost and the prover has to send all $2 \cdot 2^{k_{\mathrm{s}}} - 1$ evaluations.

\paragraph{An improved strategy.} To allow both optimizations to live side-by-side (at a minimal cost) we revisit \cref{eq:skip_c}. Two observations make this work.

\emph{First,} we keep the zerocheck saving by running just the first round of sumcheck on the RHS as well: instead of splitting $\hat{c}$ off up-front, we have the prover send the round-1 messages $P^{ab}(\lambda)$ \emph{and} $P^{c}(\lambda)$ separately, where $P^{ab}$ is the univariate skip message for the LHS and $P^{c}$ is the corresponding message for the RHS. The verifier checks that $P^{ab} - P^{c}$ vanishes on $S$ (this is the original zerocheck identity, now applied to the difference) and samples a single challenge $\lambda \in \F$. The $\hat{c}$-claim is then read off directly as $P^{c}(\lambda) = \hat{c}(\lambda, \mathbf{r})$, so $c$ still drops out of the residual sumcheck.

\emph{Second,} computing $P^{c}$ is itself cheap. At first glance, the univariate skip on the RHS would require computing the LDEs of all $2^{m - k_{\mathrm{s}}}$ rows of $c$ (viewing $c$ as a $2^{m - k_{\mathrm{s}}} \times 2^{k_{\mathrm{s}}}$ matrix over $\Ftwo$, as in \cref{sec:impl:uskip}) --- which would in fact be \emph{more} expensive than processing $c$ alongside $a, b$ in the original protocol. But by linearity, we can first aggregate the rows of $c$ under the eq-weights $\eq(\mathbf{r}, \cdot)$ to produce a single $\Fbig$-valued vector of length $2^{k_{\mathrm{s}}}$, and then take a single LDE of that vector to obtain $P^{c}$ on $\Lambda$. This makes $P^{c}$ extremely cheap to compute.

\paragraph{One more optimization.} The way we instantiate the R1CS from our circuits has a structural property: the constraint matrix $C_0$ is simply the identity matrix, so $c = z$. Thus, there is never a need to materialize $c$ separately --- the existing $z$-buffer plays the role of $c$ wherever needed. Combined with the optimizations above, the $\hat{c}$-claim $\hat{c}(\lambda, \mathbf{r})$ is a direct claim on $\hat{z}$ at the same point. Thus, the vector $c$ also does not need to be ``lin-checked''.

\begin{remark}
  One caveat of these optimizations is that the claim that we get on $\hat{c}$ is at a different point than the one we get on $\hat{a}$ and $\hat{b}$. We deal with this using the known batching techniques of~\cite{RR24,HyperPlonk23}, but note that there is a small cost associated with that.
\end{remark}

\subsection{Circuit walking}
\label{sec:impl:walker}

We focus throughout this subsection on the constraint matrix $A_0$; the same construction applies to $B_0$ (and the optimization of \cref{sec:impl:c} fixes $C_0 = I$ so no further optimizations are needed for it). In a nutshell, we describe an optimization that enables us to avoid ever materializing the matrix $A_0$ by following the underlying circuit structure.

\paragraph{R1CS for circuits, and the cost of substitution.} There are multiple ways to map a (Boolean) circuit into an R1CS instance. We highlight two extremes:
\begin{itemize}
\item \emph{All wires.} Commit every internal wire of the circuit as a separate entry of $z$. Each row of $A_0$ has $O(1)$ nonzeros (only the wires immediately feeding the gate's $A$-side), but the witness is large.
\item \emph{Only ANDs.} Commit only the AND-gate outputs (and the inputs/outputs of the whole circuit). Every internal wire is then a fixed $\F_2$-linear combination of these committed values --- evaluate the linear cascades symbolically, substitute everywhere, and the witness shrinks (often by a factor of two or more). The cost is that each row of $A_0$ now contains nonzeros at every committed column reachable, through the substitution, from the gate's $A$-side --- typically a dense linear combination whose fan-in scales with the depth of intervening linear computation.
\end{itemize}
The substituted form is much better for almost everything --- PCS, witness storage, witness generation, zerocheck. The catch is that any algorithm that needs to iterate over $A_0$'s nonzeros pays the substituted cost, which can be drastically more expensive than the all-wires equivalent. (In our actual implementation this is most evident in our Keccak arithmetization, which has an extremely dense matrix $A_0$.) Next, we show how to mitigate the cost of this representation.

\paragraph{Warmup: witness generation as a circuit walk.} In the substituted encoding above, only the AND-gate outputs are committed in $z$ (alongside the circuit's inputs and outputs); the XOR-gate outputs are intermediate wires that the prover holds only transiently.

In the protocol we need to generate the vector $a=Az$ (this is part of what is known as \emph{trace} generation). Recall that in our setting $A = I_K \otimes A_0$ and so this can be done via $K$ vector-matrix multiplications by $A_0$.

However, it is easy to see that doing so is wasteful --- a much better option is to simply evaluate the circuit, gate by gate, in topological order. By the definition of $A_0$, the $A$-side input of each AND-gate is exactly the corresponding entry of $A_0 z$, and the $B$-side input is the corresponding entry of $B_0 z$. So this forward walk \emph{is} the matrix--vector products $A_0 z$ and $B_0 z$ --- interleaved with per-gate ANDs and writes back into $z$, and computed at the cost of one $\F_2$-evaluation of the hash circuit, with no $A_0$ or $B_0$ ever materialized.

\paragraph{Lincheck's hot path.} Recall from \cref{sec:impl:batch} that after the block-diagonal collapse, the lincheck prover and verifier only ever interact with the small base matrix $A_0 : \bitset^{m_0} \times \bitset^{m_0} \to \F_2$. The lincheck identity for $A$ requires evaluating, for each Boolean column index $\mathbf{j}_{\mathrm{in}} \in \bitset^{m_0}$, the eq-weighted column marginal
\[
  \eta^A[\mathbf{j}_{\mathrm{in}}] \;:=\; \hat{A}_0(\mathbf{r}_{y,\mathrm{in}}, \mathbf{j}_{\mathrm{in}}) \;=\; \sum_{\mathbf{i}_{\mathrm{in}} \,:\, A_0[\mathbf{i}_{\mathrm{in}}, \mathbf{j}_{\mathrm{in}}] = 1} \eq(\mathbf{r}_{y,\mathrm{in}}, \mathbf{i}_{\mathrm{in}}),
\]
which is then paired against the partial fold of the witness $\hat{z}(\mathbf{r}_{y,\mathrm{out}}, \cdot)$ in an $m_0$-round multilinear sumcheck. Letting $E[\mathbf{i}_{\mathrm{in}}] := \eq(\mathbf{r}_{y,\mathrm{in}}, \mathbf{i}_{\mathrm{in}})$, this is a single matrix--vector product
\[
  \eta^A \;=\; A_0^\top \cdot E
\]
--- the same matrix $A_0$ that appeared in witness generation, just transposed and applied to $E$ in place of $z$. The naive way to compute it is to materialize $A_0$ as a sparse Boolean matrix and, for each row $\mathbf{i}_{\mathrm{in}}$, scatter $E[\mathbf{i}_{\mathrm{in}}]$ into every column appearing in that row, paying one $\F$-addition per nonzero --- in the substituted encoding, exactly the cost we want to avoid.

\paragraph{Walking the circuit backwards.} Since the forward walk produced $A_0 z$ by processing the circuit's XOR and AND gates in topological order, we can produce $A_0^\top E$ by the transposed walk: traverse the same gates in reverse, with $\F$-additions in place of $\F_2$-XORs. Concretely, maintain a running marginal $M^A$ over the circuit's wires, initialized to zero, and process the gates in reverse topological order:
\begin{itemize}
\item At an AND-gate, the gate's R1CS row carries weight $\eq(\mathbf{r}_{y,\mathrm{in}}, \cdot)$ at the row's index. Add this weight into $M^A[w]$ for each wire $w$ on the gate's $A$-side; if $w$ is committed in $z$, deposit directly into the corresponding entry of $\eta^A$.
\item At an XOR-gate $y = x_1 \oplus \cdots \oplus x_k$, add $M^A[y]$ into each of $M^A[x_1], \ldots, M^A[x_k]$. (This is the transpose of XOR: a single output weight is distributed back to all its inputs.)
\end{itemize}
When the walk reaches the circuit's inputs (which are committed in $z$), deposit the remaining $M^A$ into the corresponding entries of $\eta^A$. The result is exactly the column marginal $\eta^A$ that the sparse scatter on the substituted $A_0$ would have produced --- but no substituted matrix ever appears.

\paragraph{Cost.} Each XOR-gate's backward step distributes one $\F$-value to each of its inputs, so the backward pass over the XOR sub-circuit costs as many $\F$-additions as the forward pass costs $\F_2$-XORs during witness generation. The AND-gate work likewise mirrors the forward work --- one scatter per AND-gate. The total walker cost is therefore proportional to the cost of \emph{evaluating the underlying circuit once}, in $\F$-arithmetic instead of $\F_2$-arithmetic, and is independent of how aggressive the substitution is. The walker also keeps memory traffic small: only a small wire-window of $\F$-elements is live at any time, fitting comfortably in cache. By contrast, the sparse scatter on the substituted $A_0$ has nonzero count proportional to (circuit size) $\times$ (substitution fan-in), and at large $2^{m_0}$ its parallel form requires per-thread accumulators of length $2^{m_0}$, whose reduction is memory-bandwidth-bound.

\paragraph{Example: Keccak.} The Keccak-$f$ permutation is $24$ rounds of a $1600$-bit state, each round consisting of a linear part $\theta \circ \rho \circ \pi$ (plus the round constant) followed by $1600$ $\chi$ AND-gates --- $24 \cdot 1600 = 38{,}400$ AND-gates in total, and so (roughly) this many rows in $A_0$. The two encoding choices are very different:
\begin{itemize}
\item \emph{All wires.} Commit every intermediate state $s_0, s_1, \ldots, s_{24}$ and every $\chi$-output $t_0, \ldots, t_{23}$ --- $49 \cdot 1600 = 78{,}400$ bits per Keccak (forcing the R1CS dimension to $m_0 = 17$). Each row of $A_0$ has $O(1)$ nonzeros.
\item \emph{Only ANDs.} Drop $s_1, \ldots, s_{23}$ and commit only $s_0$, $s_{24}$, $t_0, \ldots, t_{23}$ --- $26 \cdot 1600 = 41{,}600$ bits per Keccak ($m_0 = 16$, a nearly $2\times$ witness shrink). But each row of $A_0$ now spreads through the substitution across $s_0$ and $t_0, \ldots, t_{r-1}$, and on average is several hundred times denser than the all-wires counterpart.
\end{itemize}
The circuit walker lets us keep the substituted encoding without paying the materialization cost: lincheck's hot path is computed in roughly one Keccak evaluation per call --- well under a million $\F$-additions --- instead of iterating through the tens of millions of nonzeros in the materialized substituted $A_0$.

\paragraph{Verifier symmetry.} The walker is run identically by the prover and the verifier. The verifier's call costs the same circuit-evaluation-equivalent number of field additions, which is small compared to the verifier's other per-instance work (it is also independent of the batch size $K$ thanks to the block-diagonal collapse of \cref{sec:impl:batch}).

\subsection{Input/output constraints}
\label{sec:impl:io}

So far we have only proved that each of the $K$ batched instances satisfies the same base R1CS. To express the cross-instance statement that the user actually cares about --- e.g.\ that the instances correspond to a hash chain, or consistency with Merkle paths --- we need to bind the per-instance \emph{input/output (IO) regions} of the witness to an auxiliary glue circuit $G$, as outlined in \cref{sec:intro}. This subsection describes the generic mechanism, and then the specific instantiation of a hash-chain.

\paragraph{Generic IO via slot-aligned regions.} Recall from \cref{sec:impl:batch} that the witness $\mathbf{z} : \bitset^k \times \bitset^{m_0} \to \Ftwo$ is a $K = 2^k$-fold stack of per-instance blocks of size $2^{m_0}$. We adopt a uniform layout convention: \emph{each instance reserves a fixed, byte-aligned slot for each of its IO regions} (e.g.\ the input state and the output state of a hash compression), positioned at fixed coordinates within the block. Concretely, an IO region of $2^{m_r}$ bits is placed at an aligned offset, so the corresponding sub-cube of $\mathbf{z}$ is indexed by $(\mathbf{i}_{\mathrm{out}}, \mathbf{s}, \mathbf{b})$ where $\mathbf{i}_{\mathrm{out}} \in \bitset^k$ is the instance, $\mathbf{s} \in \bitset^{m_0 - m_r}$ are the fixed (constant) high coordinates that select the slot, and $\mathbf{b} \in \bitset^{m_r}$ ranges over the bits of the region.

With this convention, any IO claim expressible as a multilinear evaluation of a region's MLE flows back to a multilinear evaluation of $\hat{\mathbf{z}}$ at a point whose high coordinates are pinned to the slot's $\mathbf{s}$. The glue circuit $G$ is in turn handled by a small auxiliary protocol whose only output is such a claim (or a batch of them), which is then folded into the PCS opening for $\hat{\mathbf{z}}$ alongside the zerocheck/lincheck claims via the standard MLE-batching technique~\cite{RR24,HyperPlonk23}. The cost of $G$ on the prover therefore reduces to (i) the auxiliary protocol itself and (ii) one additional opening claim against $\hat{\mathbf{z}}$ --- both of which we will see are essentially free for the hash-chain.

\paragraph{Hash-chain: the goal.} A \emph{hash-chain} statement asks the prover to show, for public endpoints $x_0, x_{2^n} \in \bitset^{\ell}$, that
\[
  x_{i+1} = h(x_i) \quad\text{for all } i \in \{0, 1, \ldots, 2^n - 1\},
\]
where $h$ is the per-instance hash circuit (here $\ell$ is the input/output width, e.g.\ $\ell = 1600$ for Keccak). The $2^n$ instances are already proved to internally enforce $\mathit{output}_i = h(\mathit{input}_i)$ by the base R1CS; the chain claim is the cross-instance assertion that $\mathit{output}_i = \mathit{input}_{i+1}$ for all $i < 2^n - 1$, together with the public-endpoint constraints $\mathit{input}_0 = x_0$ and $\mathit{output}_{2^n - 1} = x_{2^n}$. The prover is given the full chain $x_0, x_1, \ldots, x_{2^n}$ in the clear, so trace generation remains fully parallel across $i$.

The naive options are unattractive. \emph{Witness aliasing} (committing $x_i$ once and pointing $\mathit{input}_{i+1}$ and $\mathit{output}_i$ to the same physical entry) breaks the block-diagonal structure that drives \cref{sec:impl:batch}. A \emph{permutation argument} via grand product (e.g.\ a Plookup-style or \cite{HyperPlonk23}-style permutation check) works but is wasteful: the underlying relation is just a length-$2^n$ shift, not an arbitrary permutation. We use instead a tailored \emph{shift argument} that exploits this structure and reduces the chain to a single MLE-evaluation claim on $\hat{\mathbf{z}}$ via one short sumcheck.

\paragraph{The shift argument.} Write $\mathrm{In}(\mathbf{i})$ and $\mathrm{Out}(\mathbf{i})$ for the input and output \emph{scalars} of instance $\mathbf{i} \in \bitset^n$, obtained by collapsing the bit dimension of each region at a shared verifier-sampled random point in $\F^{m_r}$ (by Schwartz--Zippel, the bit-level chain reduces to this scalar chain with overwhelming probability). Let $\mathrm{shift}(\mathbf{a}, \mathbf{b}) : \F^n \times \F^n \to \F$ be the multilinear extension of the successor relation $\mathbf{b} = \mathbf{a} + 1$ on $n$-bit integers (i.e.\ for $\mathbf{a}, \mathbf{b} \in \bitset^n$, $\mathrm{shift}(\mathbf{a}, \mathbf{b}) = 1$ iff $\mathbf{b}$ is the integer successor of $\mathbf{a}$, and $0$ otherwise). A closed form is easily derived by splitting on the MSB and tracking the carry, and is evaluable in $O(n)$ field operations. The (interior) chain relation $\mathrm{Out}(\mathbf{i}) = \mathrm{In}(\mathbf{i} + 1)$ is then captured by the identity
\begin{equation}\label{eq:chain-shift-identity}
  \sum_{\mathbf{y} \in \bitset^n} \eq(\boldsymbol{\tau}, \mathbf{y}) \cdot \mathrm{Out}(\mathbf{y})
  \;=\; \sum_{\mathbf{y} \in \bitset^n} \mathrm{shift}(\boldsymbol{\tau}, \mathbf{y}) \cdot \mathrm{In}(\mathbf{y}),
\end{equation}
which the verifier checks via sumcheck after sampling a single challenge $\boldsymbol{\tau} \in \F^n$ --- yielding a single MLE-evaluation claim on $\hat{\mathbf{z}}$ that is batched into the PCS opening alongside the zerocheck and lincheck claims.\footnote{The public endpoint constraints $\mathrm{In}(\mathbf{0}) = \hat{x}_0$ and $\mathrm{Out}(\mathbf{1}^n) = \hat{x}_{2^n}$ are folded into \cref{eq:chain-shift-identity} by standard tricks --- the output endpoint via the boundary term that drops out of $\mathrm{shift}$ in characteristic $2$, and the input endpoint via a random linear combination --- without adding rounds or changing the sumcheck's structure. Details omitted.}

Crucially, because the input and output regions sit in two consecutive slots differing only in a single selector bit, the two MLE evaluations $\mathrm{In}(\boldsymbol{\tau}')$ and $\mathrm{Out}(\boldsymbol{\tau}')$ that the sumcheck would have produced at its challenge point $\boldsymbol{\tau}' \in \F^n$ can be merged into one MLE evaluation of $\hat{\mathbf{z}}$, by running a single extra sumcheck round over the selector bit.

\paragraph{Cost.} The chain claim's evaluation point has a particularly favorable structure: its high coordinates contain $m - n - m_r - 1$ zero entries (the fixed slot-selector bits, all but $s_0$), each of which halves the live support of the eq-weighting in the PCS opening. Leveraging this fact, the chain claim adds well under $5\%$ to the end-to-end prover wall-clock across BLAKE3, SHA-256, and Keccak.

\paragraph{Other IO circuits.} The same template applies to other small glue circuits: a single auxiliary sumcheck that reduces $G$ to one or a few $\hat{\mathbf{z}}$-evaluation claims at points whose high coordinates pin the IO slot. For instance, Merkle paths --- the other IO circuit we instantiate --- can be reduced analogously, with the auxiliary protocol expressing the parent/child relation as a sparse linear combination over the committed leaf and internal-node slots. We omit the details.

\section{Evaluation}
\label{sec:eval}
Our prototype Rust implementation of \sysname and benchmarking scripts are available online\footnote{\href{https://github.com/succinctlabs/flock}{https://github.com/succinctlabs/flock}}; see \texttt{BENCHMARKS.md} to reproduce benchmarks. The code was developed with the assistance of coding agents, specifically Claude Opus 4.7 and 4.8. The field arithmetic and ring-switching implementations are adapted from \cite{Irr}, and the Ligerito implementation is adapted from \cite{Bolt} (all with significant modifications).

We evaluate \sysname in terms of its ability to prove many independent executions of three hash function primitives: Keccak-f[1600] permutations, SHA-256 compressions, and BLAKE3 compressions.
When possible, we compare its performance to other systems:
\begin{itemize}
  \item Binius64~\cite{Irr} (pinned at \texttt{8f21b34}) for Keccak-f[1600], SHA-256, and BLAKE3. Since the original implementation does not support multi-threaded witness generation, we used coding agents to add this feature. Binius64 targets 96 bits of security.
  \item Plonky3~\cite{Plonky3} (pinned at \texttt{109e95c}) for Keccak and BLAKE3. Plonky3's default configuration targets 113 bits under proximity gap \emph{conjectures} that seem problematic following recent attacks \cite{DG25,CS25a,FS25}. That configuration only has 65 proven bits of security. We therefore adjusted the configuration to target 100 bits of proven security.
  \item Hashcaster~\cite{hashcaster} for Keccak. The original implementation does not include a PCS, so we use the implementation from~\cite{bench-hash-in-snark} (pinned at \texttt{1af6fc5}). Hashcaster targets 100 bits of security.
  \item Vega-MC from the vega-prover library~\cite{Spartan2} (pinned at \texttt{0d4f140}), for SHA-256 only. Unlike the post-quantum systems above, Vega-MC is curve-based (over the 256-bit T256 curve) and targets $\approx 128$ bits of computational, \emph{non-post-quantum} security. The library provides two provers; we report its faster \emph{VEGA-MC} folding prover, which amortizes the per-proof cost across the batch, measured at its throughput-optimal batch size ($2^{\bench{sha.nn.mt.exp}}$ compressions). Vega-MC is by far the fastest of all elliptic-curve based implementations we benchmarked (vega-sc, gnark, snarkjs, noir).
\end{itemize}
For all provers we take into account the full per-instance prover time (witness generation, commitment, and proof), matching \sysname's witness-generation-inclusive timing.

We also remark that Binius64 and Vega target general-purpose circuits, whereas Plonky3 (via precompiles), Hashcaster and \sysname are optimized for specific computations. Binius64, Vega-MC support zero-knowledge. The other systems, including \sysname, currently do not.
\paragraph{Evaluation setup.}
All benchmarks were conducted on a single \bench{meta.hardware} ($\bench{meta.ram}$\,GB RAM, $\bench{meta.cores}$ performance cores). Multi-threaded benchmarks use $\bench{meta.threads}$ threads, which matches the number of performance cores available. Our implementation is optimized for this ARM architecture, while other implementations may perform better on other hardware.
When measuring prover throughput, we take the best of three trials, following a warmup trial.

\paragraph{Security parameterization.}
We target $\bench{meta.security}$-bits of security by default (but also briefly discuss a variant targeting $120$-bits of security in \cref{tab:keccak-fixed}). This means that the underlying IOP has round-by-round soundness error \cite{CCHLRR18} at most $2^{-\bench{meta.security}}$ (unconditionally -- without relying on any unproven conjectures). We rely on the cryptographic hardness of the SHA-256 compression function for collision resistance and Fiat-Shamir security.

We utilize the field $\mathbb{F}_{2^{128}}$ throughout. This more than suffices for protocols such as sumcheck, zerocheck and ring-switching. We use Ligerito~\cite{NA25} as our PCS but extend it to the list-decoding regime (see \cref{appx:ligerito}). We set the initial code to be an (interleaved) Reed-Solomon code of rate $\rho=\frac{1}{2}$ and use distance and proximity gaps up to the Johnson bound~\cite{BCIKS20,BCHKS25} (in combination with a small amount of grinding to achieve the desired $2^{-\bench{meta.security}}$ error). One could additionally do grinding for the query phase. In our main benchmarks we do not do so, but in \cref{tab:keccak-fixed} we also give benchmarks for a ``slim'' variant of \sysname that minimizes proof size and does introduce limited query grinding.

Concretely, for proving $\approx 2^{14}$ Keccak permutations the five levels use rates $1/2, 1/4, 1/8, 1/16, 1/32$ with $218, 106, 71, 53, 43$ queries and up to $17, 12, 9, 6, 4$ bits of proximity-gap grinding, respectively. The query counts follow the formula in \cref{rem:query_count} (so they depend only on the rate), whereas the grinding scales mildly with the instance size.

\subsection{Prover throughput}
\label{sec:eval-throughput}
We measure prover throughput, i.e., the number of hash function primitives proven per second.
For each system, we benchmark a range of batch sizes $(2^{10}, 2^{12}, 2^{14}, 2^{16}, 2^{18})$, stopping once the peak memory usage roughly exceeds 10\,GB.\footnote{For Binius64, we limit the maximum batch size to be $2^{14}$, since the offline circuit builder uses prohibitively large amounts of memory past this point.}
We report the maximum prover throughput across the batch sizes, for both single-threaded and multi-threaded executions.
For reference, we also report the native throughput of executing the hash function primitives on a single core, \emph{without} specialized hardware instructions.

\paragraph{Keccak.}
\cref{tab:keccak-throughput} reports maximum proving throughputs for Keccak-f[1600] permutations.
For padding reasons, the batch sizes for \sysname and HashCaster are roughly $1.5 \times$ larger (e.g., $1.5 \cdot 2^{14}$ instead of $2^{14}$) and the batch sizes for Plonky3 are roughly $1.33 \times$ larger.
For single-threaded performance, \sysname reaches a maximum throughput of $\bench{kc.flock.st}$k permutations per second; this represents a $\bench{overhead.kc} \times$ overhead over native execution.
For multi-threaded performance, \sysname reaches a maximum throughput of $\bench{kc.flock.mt}$k permutations per second, which is roughly $\bench{sp.kc.mt.hc} \times$ faster than Hashcaster, $\bench{sp.kc.mt.b64} \times$ faster than Binius64, and $\bench{sp.kc.mt.p3} \times$ faster than Plonky3.

\begin{table}[t]
  \centering
  \begin{tabular}{lrr}
    \toprule
    & \multicolumn{2}{c}{max throughput \textbf{keccak/s}} \\
    \cmidrule(lr){2-3}
    system & \multicolumn{1}{c}{single-thread} & \multicolumn{1}{c}{multi-thread} \\
    \midrule
    \sysname   & $\bench{kc.flock.st}\text{k}_{\,(2^{\bench{kc.flock.st.exp}})}$ & $\bench{kc.flock.mt}\text{k}_{\,(2^{\bench{kc.flock.mt.exp}})}$ \\
    Hashcaster   & $\bench{kc.hc.st}\text{k}_{\,(2^{\bench{kc.hc.st.exp}})}$ & $\bench{kc.hc.mt}\text{k}_{\,(2^{\bench{kc.hc.mt.exp}})}$ \\
    Binius64     & $\bench{kc.b64.st}\text{k}_{\,(2^{\bench{kc.b64.st.exp}})}$  & $\bench{kc.b64.mt}\text{k}_{\,(2^{\bench{kc.b64.mt.exp}})}$ \\
    Plonky3      & $\bench{kc.p3.st}\text{k}_{\,(2^{\bench{kc.p3.st.exp}})}$ & $\bench{kc.p3.mt}\text{k}_{\,(2^{\bench{kc.p3.mt.exp}})}$  \\
    \midrule
    Native & \multicolumn{1}{c}{$\bench{kc.native}\text{M}$} \\
    \bottomrule
  \end{tabular}
  \caption{Proving vs.\ native computing throughputs for Keccak-f[1600] permutations.
    Subscripts indicate the approximate number of permutations at which proving throughput is maximized.}
  \label{tab:keccak-throughput}
\end{table}

\paragraph{SHA-256.}
\cref{tab:sha256-throughput} reports maximum proving throughputs for SHA-256 compressions.
For single-threaded performance, \sysname reaches a maximum throughput of $\bench{sha.flock.st}$k compressions per second; this represents a $\bench{overhead.sha} \times$ overhead over native execution.
For multi-threaded performance, \sysname reaches a maximum throughput of $\bench{sha.flock.mt}$k compressions per second, which is roughly $\bench{sp.sha.mt.b64} \times$ faster than Binius64 and $\bench{sp.sha.mt.nn} \times$ faster than Vega-MC.

\begin{table}[t]
  \centering
  \begin{tabular}{lrr}
    \toprule
    & \multicolumn{2}{c}{max throughput, \textbf{sha-256/s}} \\
    \cmidrule(lr){2-3}
    system & \multicolumn{1}{c}{single-thread} & \multicolumn{1}{c}{multi-thread} \\
    \midrule
    \sysname   & $\bench{sha.flock.st}\text{k}_{\,(2^{\bench{sha.flock.st.exp}})}$ & $\bench{sha.flock.mt}\text{k}_{\,(2^{\bench{sha.flock.mt.exp}})}$ \\
    Binius64     & $\bench{sha.b64.st}\text{k}_{\,(2^{\bench{sha.b64.st.exp}})}$  & $\bench{sha.b64.mt}\text{k}_{\,(2^{\bench{sha.b64.mt.exp}})}$ \\
    Vega-MC  & $\bench{sha.nn.st}\text{k}_{\,(2^{\bench{sha.nn.st.exp}})}$  & $\bench{sha.nn.mt}\text{k}_{\,(2^{\bench{sha.nn.mt.exp}})}$ \\
    \midrule
    Native & \multicolumn{1}{c}{$\bench{sha.native}\text{M}$} \\
    \bottomrule
  \end{tabular}
    \caption{Proving vs.\ native computing throughputs for SHA-256 compressions.
    Subscripts indicate the approximate number of compressions at which proving throughput is maximized.}
  \label{tab:sha256-throughput}
\end{table}

\paragraph{BLAKE3.}
\cref{tab:blake3-throughput} reports maximum proving throughputs for BLAKE3 compressions.
For single-threaded performance, \sysname reaches a maximum throughput of $\bench{bl.flock.st}$k compressions per second; this represents a $\bench{overhead.bl} \times$ overhead over native execution.
For multi-threaded performance, \sysname reaches a maximum throughput of $\bench{bl.flock.mt}$k compressions per second, which is roughly $\bench{sp.bl.mt.head} \times$ faster than Binius64 and Plonky3.

S-two\footnote{Available at \url{https://github.com/starkware-libs/stwo}}\cite{STWO} has benchmarks only for BLAKE2s and single-threaded. On the same machine it proves $\bench{stwo.bl2s.st}\text{k}$ BLAKE2s compressions per second. Since BLAKE3 has only 7 rounds as compared to BLAKE2s which has 10, we extrapolate that S-two should be able to prove about $\bench{stwo.bl3.st}\text{k}$ BLAKE3 compressions per second -- about $\bench{sp.bl.st.stwo}\times$ slower than Flock.

\begin{table}[t]
  \centering
  \begin{tabular}{lrr}
    \toprule
    & \multicolumn{2}{c}{max throughput \textbf{blake3/s}} \\
    \cmidrule(lr){2-3}
    system & \multicolumn{1}{c}{single-thread} & \multicolumn{1}{c}{multi-thread} \\
    \midrule
    \sysname   & $\bench{bl.flock.st}\text{k}_{\,(2^{\bench{bl.flock.st.exp}})}$ & $\bench{bl.flock.mt}\text{k}_{\,(2^{\bench{bl.flock.mt.exp}})}$ \\
    Binius64     & $\bench{bl.b64.st}\text{k}_{\,(2^{\bench{bl.b64.st.exp}})}$  & $\bench{bl.b64.mt}\text{k}_{\,(2^{\bench{bl.b64.mt.exp}})}$ \\
    Plonky3      & $\bench{bl.p3.st}\text{k}_{\,(2^{\bench{bl.p3.st.exp}})}$  & $\bench{bl.p3.mt}\text{k}_{\,(2^{\bench{bl.p3.mt.exp}})}$ \\
    \midrule
    Native & \multicolumn{1}{c}{$\bench{bl.native}\text{M}$} \\
    \bottomrule
  \end{tabular}
    \caption{Proving vs.\ native computing throughputs for BLAKE3 compressions.
    Subscripts indicate the approximate number of compressions at which proving throughput is maximized.}
  \label{tab:blake3-throughput}
\end{table}

\subsection{Additional benchmarks}

\Cref{fig:blake3-scaling} plots \sysname's BLAKE3 proving throughput at different batch sizes, both single-threaded and multi-threaded.
Single-threaded throughput saturates at $2^{\bench{bl.flock.st.exp}}$ compressions, whereas multi-threaded throughput continues to rise, roughly approaching an $\bench{scaling.speedup.max} \times$ speedup.

\begin{figure}[t]
  \centering
  \begin{tikzpicture}
    \begin{axis}[
      width=0.9\columnwidth,
      height=0.66\columnwidth,
      xlabel={batch size (\#BLAKE3)},
      xtick={10,12,14,16,18},
      xticklabels={$2^{10}$,$2^{12}$,$2^{14}$,$2^{16}$,$2^{18}$},
      xmin=9.3, xmax=18.7,
      axis y line*=left,
      ymin=0, ymax=\bench{scaling.ymax.st},
      ylabel={single-thread (k BLAKE3/s)},
      grid=major,
    ]
      \addplot[very thick, mark=*, mark size=2pt, color=blue!70!black]
        coordinates {\blakeScalingST};
    \end{axis}
    \begin{axis}[
      width=0.9\columnwidth,
      height=0.66\columnwidth,
      xmin=9.3, xmax=18.7,
      hide x axis,
      axis y line*=right,
      ymin=0, ymax=\bench{scaling.ymax.mt},
      ylabel={multi-thread (k BLAKE3/s)},
    ]
      \addplot[very thick, mark=square*, mark size=2pt, color=red!70!black]
        coordinates {\blakeScalingMT};
      \blakeScalingLabels
    \end{axis}
  \end{tikzpicture}
  \caption{\sysname BLAKE3 proving throughput vs.\ batch size. Blue (left
    axis) is single-threaded; red (right axis) is multi-threaded (10 threads). Each multi-threaded point is labeled with its speedup over the
    single-threaded prover at the same batch size.}
  \label{fig:blake3-scaling}
\end{figure}
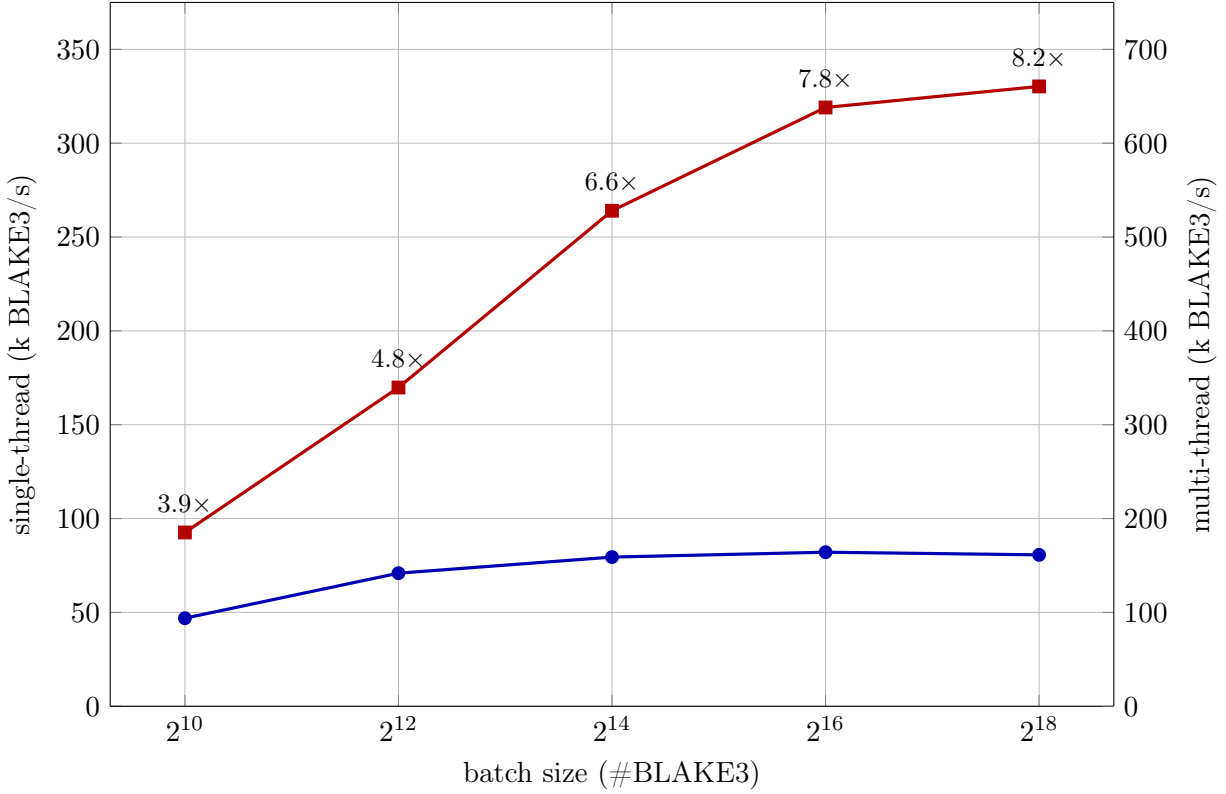

\cref{tab:keccak-fixed} fixes the batch size to be $\approx 2^{\bench{kcfix.exp}}$ Keccak permutations and reports the resulting prover throughput, proof size, and verifier time.
We include three variants of \sysname: \emph{fast} (the default used elsewhere, PCS rate $1/2$, list-decoding regime), \emph{slim} (rate $1/4$, which trades prover throughput for smaller proofs), and \emph{120-security} (rate $1/2$ in the unique-decoding regime, targeting $120$-bit rather than $100$-bit security).
The exact number of permutations ($1.5 \cdot 2^{\bench{kcfix.exp}}$ for \sysname and Hashcaster, $2^{\bench{kcfix.exp}}$ for Binius64, and $1.33 \cdot 2^{\bench{kcfix.exp}}$ for Plonky3) depends on padding.
\sysname leads on throughput, produces the smallest proofs ($\bench{kcfix.flock.fast.proof}$\,KiB for the fast variant, $\bench{kcfix.flock.slim.proof}$\,KiB for the slim variant), and has the fastest verifier (around $\bench{kcfix.flock.verify.approx}$\,ms, single-threaded, across all variants).

\begin{table}[!ht]
  \centering
  \begin{tabular}{lrrrr}
    \toprule
    system & m/t throughput & proof & verify \\
    \midrule
    \sysname (fast)   & $\bench{kcfix.flock.fast.thru}$k/s & $\bench{kcfix.flock.fast.proof}$\,KiB  & $\bench{kcfix.flock.fast.verify}$\,ms \\
    \sysname (slim)   & $\bench{kcfix.flock.slim.thru}$k/s & $\bench{kcfix.flock.slim.proof}$\,KiB  & $\bench{kcfix.flock.slim.verify}$\,ms \\
    \sysname (120-security) & $\bench{kcfix.flock.sec.thru}$k/s & $\bench{kcfix.flock.sec.proof}$\,KiB  & $\bench{kcfix.flock.sec.verify}$\,ms \\
    Hashcaster   & $\bench{kcfix.hc.thru}$k/s & $\bench{kcfix.hc.proof}$\,KiB  & $\bench{kcfix.hc.verify}$\,ms \\
    Binius64     & $\bench{kcfix.b64.thru}$k/s & $\bench{kcfix.b64.proof}$\,KiB  & $\bench{kcfix.b64.verify}$\,ms \\
    Plonky3      & $\bench{kcfix.p3.thru}$k/s & $\bench{kcfix.p3.proof}$\,MiB & $\bench{kcfix.p3.verify}$\,ms \\
    \bottomrule
  \end{tabular}
  \caption{Throughputs, proof sizes, and verifier times for proving $\approx 2^{\bench{kcfix.exp}}$ Keccak-f[1600] permutations.}
  \label{tab:keccak-fixed}
\end{table}

\cref{tab:phase-breakdown} reports \sysname's cost breakdowns for proving $\approx 2^{\bench{kcfix.exp}}$ Keccak-f[1600] permutations, SHA-256 compressions, and BLAKE3 compressions, as a percentage of overall prover time.
The shape is the same across all three hash function primitives.
The PCS and zerocheck phases dominate, accounting for $\bench{phase.pcszero.lo}$--$\bench{phase.pcszero.hi}\%$ of prover time, while witness generation and lincheck are comparatively cheap.

\begin{table}[!ht]
  \centering
  \begin{tabular}{lrrr}
    \toprule
    phase & Keccak & SHA-256 & BLAKE3 \\
    \midrule
    witness gen   & $\bench{phase.kc.witness_gen}\%$  & $\bench{phase.sha.witness_gen}\%$  & $\bench{phase.bl.witness_gen}\%$  \\
    PCS commit    & $\bench{phase.kc.pcs_commit}\%$ & $\bench{phase.sha.pcs_commit}\%$ & $\bench{phase.bl.pcs_commit}\%$ \\
    zerocheck     & $\bench{phase.kc.zerocheck}\%$ & $\bench{phase.sha.zerocheck}\%$ & $\bench{phase.bl.zerocheck}\%$ \\
    lincheck      & $\bench{phase.kc.lincheck}\%$  & $\bench{phase.sha.lincheck}\%$  & $\bench{phase.bl.lincheck}\%$  \\
    PCS open      & $\bench{phase.kc.pcs_open}\%$ & $\bench{phase.sha.pcs_open}\%$ & $\bench{phase.bl.pcs_open}\%$ \\
    \bottomrule
  \end{tabular}
    \caption{\sysname's cost breakdown for proving
    $\approx 2^{\bench{kcfix.exp}}$ hash function primitives.}
  \label{tab:phase-breakdown}
\end{table}
\FloatBarrier

\ifsubmission\else
  \iffull
    \section*{Acknowledgments}
  \else
    \ifCLASSOPTIONcompsoc
      \section*{Acknowledgments}
    \else
      \section*{Acknowledgment}
    \fi
  \fi
  We thank Eli Ben-Sasson, Tamir Hemo, Lev Soukhanov, and Srinath Setty for useful discussions.
\fi

\bibliographystyle{alpha}
\bibliography{refs}

\iffull
  \appendix
\else
  \appendices
\fi
\crefalias{section}{appendix}
\crefalias{subsection}{appendix}
\crefalias{subsubsection}{appendix}
\crefname{appendix}{Appendix}{Appendices}
\Crefname{appendix}{Appendix}{Appendices}

\newcommand{\SkipParam}{k_{\mathrm{s}}}

\section{Full protocol}
\label{app:full-protocol}

We describe the \sysname protocol with the optimizations discussed in \cref{sec:impl}.
We focus on describing the verifier, since only its behavior can affect soundness.

\paragraph{Setup.}
Let $\F = \Fbig$.
Fix base matrices $A_0, B_0 \in \F_2^{2^{m_{0}} \times 2^{m_{0}}}$, and let $K = 2^k$ be the number of instances.
We assume $m_0 \geq 6$.
The batched matrices are the block-diagonal
\begin{equation*}
  A = I_K \otimes A_0 \enspace, \quad B = I_K \otimes B_0
\end{equation*}
square matrices in $\F_2^{2^{m} \times 2^{m}}$ where $m \coloneqq k + m_0$.

Recall that the prover's goal is to prove, for some assignment $\mathbf{z} \in \F^{2^{m}}$, that $(A\mathbf{z}) \circ (B\mathbf{z}) = C\mathbf{z}$ and moreover $G(\mathbf{z}) = 0$ for some auxiliary input/output circuit $G$.

\paragraph{Quirky extensions.}
Fix $\SkipParam = 6$ and a subset $S \subseteq \F$ of size $2^{\SkipParam}$.\footnote{The specific choice of $\SkipParam$ is flexible, but we pick $6$ for efficiency.}
Set $D_{0} \coloneqq \bitset^{m_{0} - \SkipParam} \times S$ and $D \coloneqq \bitset^{m - \SkipParam} \times S$. We view $\mathbf{z}, \mathbf{a}, \mathbf{b}$ as functions $z, a, b \colon D \to \F_2$.
We write $\hat{z} \colon \F^{m - \SkipParam + 1} \to \F$ to denote the \emph{quirky extension}\footnote{These types of extensions have different names in the literature, e.g., \emph{oblong extensions} \cite{Irr} and \emph{prismalinear extensions} \cite{SWIRL}.} of $z$ over $\F$, namely, the unique multivariate polynomial that is linear in the first $m - \SkipParam$ variables, is of degree less than $2^{\SkipParam}$ in the last variable, and agrees with $z$ over $D$.
We also write $\hat{a}, \hat{b}$ to denote the quirky extensions of $a, b$ over $\F$.

Similarly, we view $A_{0}, B_{0}$ as functions mapping from $D_{0} \times D_{0}$ to $\F$, and we write $\hat{A}_{0}, \hat{B}_{0}$ to denote their quirky extensions over $\F$.

\paragraph{Quirky PCS.}
We rely on a Boolean \emph{quirky} PCS, which supports committing to quirky extensions of Boolean-valued functions $f: \bitset^{m} \to \bitset$.

\paragraph{Protocol.} The full protocol proceeds as follows.

\begin{enumerate}
\item \textbf{Commit to the trace.} The prover commits to $\hat{z}$ using a Boolean quirky PCS (without out-of-domain sampling, see \cref{remark:no-ood}), producing a commitment transcript $\mathsf{cm}$.

\item \textbf{Zerocheck.}
The initial claim is that
\begin{equation*}
  a(\mathbf{x}, y) \cdot b(\mathbf{x}, y) - z(\mathbf{x}, y) = 0 \quad\text{for all } (\mathbf{x}, y) \in D \enspace.
\end{equation*}

\emph{Friendly challenges} (\cref{sec:impl:friendly}). The verifier chooses $\mathbf{r} \in \F^{m - \SkipParam}$ so that the first $7$ coordinates are fixed constants (which are $\F_2$-linearly independent, see \cref{sec:impl:friendly}), and the remaining coordinates are randomly sampled.
The resulting claim is that
\begin{equation*}
  \sum_{\mathbf{x} \in \bitset^{m-\SkipParam}} \eq(\mathbf{r}, \mathbf{x}) \cdot \big(a(\mathbf{x}, y) \cdot b(\mathbf{x}, y) - z(\mathbf{x}, y)\big) = 0 \quad \text{for all } y \in S \enspace.
\end{equation*}

\emph{Univariate skip} (\cref{sec:impl:uskip}). The prover sends polynomials $P^{ab}, P^{z}$ claimed to be
\begin{align*}
  P^{ab}(Y) & = \sum_{\mathbf{x} \in \bitset^{m-\SkipParam}} \eq(\mathbf{r}, \mathbf{x}) \cdot \hat{a}_\mathbf{x}(Y) \cdot \hat{b}_\mathbf{x}(Y) \enspace, \\
  P^{z}(Y) & = \sum_{\mathbf{x} \in \bitset^{m-\SkipParam}} \eq(\mathbf{r}, \mathbf{x}) \cdot \hat{z}_\mathbf{x}(Y) = \hat{z}(\mathbf{r}, Y) \enspace,
\end{align*}
where $\hat{a}_\mathbf{x}, \hat{b}_\mathbf{x}, \hat{z}_{\mathbf{x}}$ are the low-degree univariate extensions of the functions $a(\mathbf{x}, \cdot), b(\mathbf{x}, \cdot), z(\mathbf{x}, \cdot)$, respectively.
Observe that the degree of $P^{ab}$ is at most $2 \cdot (2^{\SkipParam} - 1)$, and the degree of $P^{z}$ is at most $2^{\SkipParam} - 1$. Since $P^{ab} - P^{z}$ is supposed to vanish on $S$, it suffices for the prover to specify $P^{ab}, P^{z}$ by sending their evaluations on a disjoint subset $\Lambda \subset \F \backslash S$ of size $2^{\SkipParam}$.
The verifier samples $\lambda \gets \F$.
The resulting claims are that
\begin{align}
  P^{ab}(\lambda) & = \sum_{\mathbf{x} \in \bitset^{m-\SkipParam}} \eq(\mathbf{r}, \mathbf{x}) \cdot \hat{a}_\mathbf{x}(\lambda) \cdot \hat{b}_\mathbf{x}(\lambda) \enspace, \label{eq:full-protocol-zerocheck} \\
  P^{z}(\lambda) & = \hat{z}(\mathbf{r}, \lambda) \enspace, \label{eq:full-protocol-c-check}
\end{align}
where $P^{ab}(\lambda)$ and $P^{z}(\lambda)$ are computed by the verifier.
Observe that \cref{eq:full-protocol-c-check} is an evaluation claim $v^{\dagger} \coloneqq P^{z}(\lambda)$ for $\hat{z}$ at the point $\mathbf{r}^{\dagger} \coloneqq (\mathbf{r}, \lambda)$; this will be handled directly in Step \ref{item:full-protocol-pcs-open}.

For the claim in \cref{eq:full-protocol-zerocheck}, the prover and verifier run an $(m-\SkipParam)$-round sumcheck, resulting in a random point $\mathbf{s}' \in \F^{m - \SkipParam}$ along with evaluation claims
\begin{equation*}
  v_{a} \overset{?}{=} \hat{a}(\mathbf{s}) \enspace, \quad v_{b} \overset{?}{=} \hat{b}(\mathbf{s}) \quad \text{ for } \mathbf{s} \coloneqq (\mathbf{s}', \lambda) \enspace.
\end{equation*}
We split $\mathbf{s} = (\mathbf{s}_{\mathrm{out}}, \mathbf{s}_{\mathrm{in}})$; note that $\mathbf{s}_{\mathrm{out}} \in \F^k$ does not contain $\lambda$ (since $\SkipParam \leq m_{0}$).

\item \textbf{Lincheck.} We reduce the claims $v_a, v_b$ to a single claim on $\hat{z}$.
We describe the reduction for $\hat{a}$ first.
The claim is that
\begin{equation*}
  v_{a} \overset{?}{=} \hat{a}(\mathbf{s}) = \sum_{(\mathbf{j}, y) \in D_{0}} \hat{A}_0(\mathbf{s}_{\mathrm{in}}, \mathbf{j}, y) \cdot \hat{z}(\mathbf{s}_{\mathrm{out}}, \mathbf{j}, y) \enspace.
\end{equation*}
The prover and verifier run an $(m_0 - \SkipParam + 1)$-round sumcheck on this claim, resulting in a random point $\mathbf{t} \in \F^{m_0 - \SkipParam + 1}$ along with evaluation claims for $\hat{A}_0(\mathbf{s}_{\mathrm{in}}, \mathbf{t})$ and $\hat{z}(\mathbf{s}_{\mathrm{out}}, \mathbf{t})$.

The same applies with $B_0$ in place of $A_0$, and batching the linchecks via \cref{rem:batch-sumcheck} reduces to a single claim
\begin{equation*}
  v^{\ddagger} \overset{?}{=} \hat{z}(\mathbf{r}^{\ddagger}) \quad \text{ for } \mathbf{r}^{\ddagger} \coloneqq (\mathbf{s}_{\mathrm{out}}, \mathbf{t}) \enspace,
\end{equation*}
together with evaluation claims on $\hat{A}_0$ and $\hat{B}_0$ at $(\mathbf{s}_{\mathrm{in}}, \mathbf{t})$ that the verifier directly checks.

\item \textbf{Consistency with auxiliary circuit.} The prover and verifier run an auxiliary protocol, which reduces the claim that $G(\mathbf{z}) = 0$ into the claim that $\hat{z}(\mathbf{r}^{\star}) = v^{\star}$ (both $\mathbf{r}^{\star}$ and $v^{\star}$ are known to the verifier).

\item \label{item:full-protocol-pcs-open} \textbf{Open the trace.} The prover and verifier run the PCS opening phase on $\mathsf{cm}$ with the evaluation claims $(\mathbf{r}^{\dagger}, v^{\dagger})$, $(\mathbf{r}^{\ddagger}, v^{\ddagger})$, and $(\mathbf{r}^{\star}, v^{\star})$. The verifier accepts if and only if the PCS verifier accepts and every check above passes.
\end{enumerate}

\paragraph{Security.}
Let $L_{\max}$ be the upper bound on the list size from~\cref{eq:johnson-bound} (or $L_{\max} = 1$ for the unique decoding regime).
The round-by-round soundness errors are:
\begin{itemize}
  \item \textbf{Commit to the trace:} $0$, since there is no out-of-domain sample.
  \item \textbf{Zerocheck:} The initial round's soundness error is $L_{\max} \cdot \frac{m - \SkipParam - 7}{|\F|}$, the univariate skip round's soundness error is $L_{\max} \cdot \frac{2 \cdot (2^{\SkipParam} - 1)}{|\F|}$, and the remaining rounds have soundness error $L_{\max} \cdot \frac{2}{|\F|}$.
  \item \textbf{Lincheck:} The univariate skip round's soundness error is $L_{\max} \cdot \frac{2 \cdot (2^{\SkipParam} - 1)}{|\F|}$, and the remaining rounds have soundness error $L_{\max} \cdot \frac{2}{|\F|}$.
  \item \textbf{Consistency with auxiliary circuit:} according to the auxiliary protocol, times $L_{\max}$.
  \item \textbf{Open the trace:} discussed in Appendix~\ref{appx:ligerito}.
\end{itemize}

\newcommand{\hi}{\mathrm{hi}}
\newcommand{\lo}{\mathrm{lo}}
\newcommand{\pkd}{\mathrm{pkd}}
\newcommand{\rA}{\mathcal{A}}        
\newcommand{\vph}{\varphi}

\section{Overview of ring-switching}\label{appx:ring_switch}
In multilinear proof-systems operating over binary fields, such as \sysname, it is convenient to arithmetize while assuming access to the underlying bits of the witness. In contrast, the underlying commitment schemes are typically designed to commit to data represented over a moderately large field.

The \emph{ring-switching} technique, introduced by Diamond and Posen \cite{biniusml} (building also on \cite{binius}), converts a polynomial commitment scheme designed to work over a large extension field into one that works over the base field, and in particular over bits. In this section we give an overview of their technique, while replacing one of their steps with a more modular, and arguably simpler, approach.

For sake of simplicity we restrict our attention to the $128$-bit extension field $\F_{2^{128}}$ but note that the same approach can be generalized to any power-of-two extension.

\subsection{Setup}
We use $\F_2$ to denote the two-element binary field and $\F = \F_{2^{128}}$ to denote its degree-$128$ extension. Recall that $\F$ is a $128$-dimensional vector space over $\F_2$. Let $\mathbf{b} = \{b_v\}_{v \in \bitset^7}$ denote a basis for $\F$ over $\F_2$; we view the $b_v$'s as elements of $\F$.

Let $q : \bitset^m \to \bitset$ be a Boolean-valued function. Let $q_{\pkd} : \bitset^{m-7} \to \F$ denote the packing of $q$ relative to $\mathbf{b}$: for $y \in \bitset^{m-7}$,
\[
  q_{\pkd}(y) = \sum_{v \in \bitset^7} q(y,v) \cdot b_v .
\]
Note that $q$ and $q_{\pkd}$ carry \emph{the same information}: $q_{\pkd}$ replaces each $128$-bit chunk of $q$'s truth table by the single $\F$-element it encodes. Hence committing to $q_{\pkd}$ via a large-field scheme incurs no ``embedding overhead''.

Let $\hat{q} : \F^m \to \F$ and $\hat{q}_{\pkd} : \F^{m-7} \to \F$ denote the multilinear extensions of $q$ and $q_{\pkd}$, respectively. The ring-switch protocol is an interactive reduction from a multilinear evaluation claim $\hat{q}(r)=\alpha$, for $r \in \F^m$ and $\alpha \in \F$, to a multilinear evaluation claim on the packed polynomial $\hat{q}_{\pkd}(r')=\alpha'$, for some $r' \in \F^{m-7}$ and $\alpha' \in \F$.

\subsection{The ring-switch protocol}

\paragraph{Decomposing the claim.}
Decompose $r \in \F^m$ as $r = (r_{\hi}, r_{\lo}) \in \F^{m-7} \times \F^7$, where $r_\lo$ corresponds to the $7$ packed coordinates. Observe that,
\begin{align}\label{eq:partial}
  \alpha \;=\; \hat q(r_\hi, r_\lo) \;=\; \sum_{v \in \bitset^7} \eq(r_{\lo}, v)\cdot \hat{q}(r_\hi, v).
\end{align}
For $v \in \bitset^7$, define the \emph{partial evaluation}:
\[
  s_v \;:=\; \hat q(r_\hi, v) \;\in\; \F.
\]
The prover generates and sends to the verifier the sequence $(s_v)_{v \in \bitset^7}$. It therefore suffices for the prover to convince the verifier that the values $(s_v)_{v \in \bitset^7}$ are correct; assuming they are, the verifier can check that $\alpha \stackrel{?}{=} \sum_v \eq(r_\lo, v) \cdot s_v$ via \cref{eq:partial}. Thus, the key task is certifying the $s_v$'s.

\begin{remark}[Why Naive Recombination is Insecure]
It is tempting to try to certify the $s_v$'s in one shot. By definition of packing,
\[
  \sum_{v \in \bitset^7} s_v\, b_v \;=\; \sum_{v} \hat{q}(r_\hi, v) \cdot b_v \;=\; \hat q_{\pkd}(r_\hi),
\]
so the verifier could obtain $\hat q_{\pkd}(r_\hi)$ from one opening of the committed packed polynomial and simply compare. This is complete but \emph{not sound}.

In more detail, the basis $\{b_v\}$ is linearly independent over $\F_2$, but \emph{not} over $\F$. Writing $\delta_v := s_v - \hat q(r_\hi,v) \in \F$ for the prover's errors, the check above amounts to the single $\F$-equation $\sum_v \delta_v\, b_v = 0$, which has an enormous space of nonzero solutions $\delta \in \F^{128}$. Adjoining the verifier's other linear check $\sum_v \delta_v \eq(r_\lo,v)=0$ leaves two $\F$-equations constraining $128$ unknowns in $\F$. A cheating prover has ample room to send false $s_v$ that pass. Abstractly, the recombination map $\sum_v s_v\, b_v$ is the multiplication map $h:\F\otimes_{\F_2}\F \to \F$, which is not injective.
\end{remark}

\paragraph{The fix: descend to $\F_2$, recombine slice-wise.}
At a high-level, the fix is to perform the linear combination \emph{where the basis is actually independent} --- over $\F_2$ --- and only then lift back to $\F$.

As $\mathbf{b}$ is a basis, we can decompose each partial evaluation $s_v$ over $\F_2$ as $s_v = \sum_{u \in \bitset^7} s_{u,v} \cdot b_u$ for $s_{u,v} \in \F_2$. Likewise decompose the equality weights: there exists a function $A : \bitset^{m-7} \times \bitset^7 \to \F_2$ such that for each $y \in \bitset^{m-7}$:
\begin{align}\label{eq:Adef}
  \eq(r_\hi, y) = \sum_{u \in \bitset^7} A(y,u) \cdot b_u.
\end{align}
Substituting \cref{eq:Adef} into the definition of $s_v$
we have that $s_v = \sum_{u \in \bitset^7} b_u \cdot \left( \sum_y  A(y,u) \cdot q(y,v) \right)$. As the inner sum is strictly over $\F_2$, and the $b_u$'s are linearly independent, each one of the summands must match the corresponding $s_{u,v}$. Thus, we have that for all $u,v \in \bitset^7$,
\begin{align}\label{eq:F2claims}
  s_{u,v} = \sum_{y \in \bitset^{m-7}} A(y,u) \cdot q(y,v).
\end{align}
These claims are \emph{purely over $\F_2$}. We may therefore safely recombine them \emph{over $v$} using the basis $\mathbf{b}$. Setting $s_u := \sum_{v} s_{u,v}\, b_v$ we have,
\begin{align}\label{eq:rowclaims}
  s_u
  \;=\; \sum_v \Big(\sum_y A(y,u)\, q(y,v)\Big) \cdot b_v
  \;=\; \sum_{y} A(y,u) \sum_v q(y,v) \cdot b_v
  \;=\; \sum_{y \in \bitset^{m-7}} A(y,u) \cdot q_{\pkd}(y),
\end{align}
one claim for each $u \in \bitset^7$. The right-hand side depends only on the packed polynomial $q_{\pkd}$ and on the publicly determined coefficients $A(y,u)$.

To see that soundness holds, assume one of the $s_v$'s sent by the prover was incorrect. By the uniqueness of the $\F_2$-decomposition, $s_{u,v}$ is then incorrect for at least one $u \in \bitset^7$. For this $u$, the recombined value $s_u = \sum_v s_{u,v} \cdot b_v$ must also be incorrect: by the $\F_2$-linear independence of the $b_v$'s, the recombination is injective in the $\F_2$-coefficients $(s_{u,v})_v$.

\subsubsection{Batching and the Sumcheck}
There are $128$ claims captured by \cref{eq:rowclaims}. The verifier batches them: it samples $r'' \leftarrow \F^7$ and computes by itself $\beta_0 = \sum_{u \in \bitset^7} \eq(r'',u) \cdot s_u$. What is left is to check that
\[
  \sum_{y \in \bitset^{m-7}} B(y)\cdot q_{\pkd}(y) = \beta_0,
\]
where $B(y) := \sum_{u \in \bitset^7} \eq(r'',u) \cdot A(y,u)$.\footnote{This batching incurs soundness error $7/|\F|$ by Schwartz--Zippel, which is negligible.}

The prover and verifier now run the standard sumcheck on $\sum_{y \in \bitset^{m-7}} B(y) \cdot \hat q_{\pkd}(y)$. The result of the sumcheck is a claim on $\hat{B}(r')$ and one on $\hat{q}_{\pkd}(r')$ for some $r' \in \F^{m-7}$. The claim on $\hat{q}_{\pkd}(r')$ is the output of the ring-switching reduction. As for the claim on $\hat B(r')$, we next show that the verifier can compute it by itself.


\paragraph{Evaluating $\hat{B}$.}

Diamond and Posen~\cite{biniusml} give a direct way to evaluate $\hat{B}$ via the perspective of tensor algebras. We give here a slightly more modular approach, which we find to be conceptually simpler.

To show how the verifier can compute $\hat{B}$, we start by giving a closed form expression for $B$. Let $\Phi : \F \to \F$ be the $\F_2$-linear map determined by $b_u \mapsto \eq(r'', u)$ on the basis $\mathbf{b}$.
\begin{claim}
  For all $y \in \bitset^{m-7}$:
  \[  B(y) = \Phi\big(\eq(r_\hi, y)\big). \]
\end{claim}
\begin{proof}
\[
  \Phi\big(\eq(r_\hi, y)\big)
  \;=\; \sum_{u \in \bitset^7} A(y, u) \cdot \Phi(b_u)
  \;=\; \sum_{u \in \bitset^7} A(y, u) \cdot \eq(r'', u)
  \;=\; B(y),
\]
where the first equality is by the $\F_2$ linearity of $\Phi$ and since $A$ takes values in $\F_2$.
\end{proof}

The equality polynomial factors over $\F$ as
\[
  \eq(r_\hi, y) \;=\; \prod_{i} g_i(y_i),
\]
where $g_i : \bitset \to \F$ is defined as $g_i(0) = 1-r_{\hi,i}$ and $g_i(1)=r_{\hi,i}$. Recall that multiplication by a fixed element of $\F$ is an $\F_2$-linear map on $\F \cong \F_2^{128}$ (via the basis $\mathbf{b}$), and so can be expressed as a $128 \times 128$ matrix over $\F_2$. Hence, for each $i$ there are two such matrices, one for each value of $y_i \in \bitset$. Thus, $B(y) = \Phi(\eq(r_\hi, y))$, evaluated on the Boolean hypercube, can be expressed as a width-$128$, length-$(m-7)$ \emph{matrix branching program} in $y$: the program reads the bits of $y$ in sequence and multiplies a $128$-bit vector by the matrix selected at each layer by the corresponding bit. After reading the last bit it further applies the linear transformation $\Phi$.

As shown by Holmgren and Rothblum \cite{HR18} (cf. \cite[Lemma 4.1]{HJRRR26} and \cite{Bloemen25}), if a function can be computed by a small width (matrix) branching program then there is an efficient algorithm to compute its multilinear extension. This yields an efficient algorithm for evaluating $\hat{B}$, as desired. For completeness, we reproduce the \cite{HR18} result in \cref{app:mbp-mle}.

\subsection{Ring-switching a quirky claim}
\label{appx:ring_switch_quirky}

The reduction above assumes a \emph{multilinear} input claim $\hat q(r) = \alpha$. In \sysname, however --- as in Binius64 --- ring-switching is invoked on the zerocheck's output, which uses the univariate-skip optimization (\cref{sec:impl:uskip}). The committed witness is then a \emph{quirky extension} (see \cref{app:full-protocol}): linear in all but the last variable, in which it has degree $< 2^{\SkipParam}$. We now explain how to extend ring-switching to handle such a claim.

The structure of the packed coordinates enters the reduction at the decomposition step (see \cref{eq:partial}) and its corresponding verifier check $\alpha \stackrel{?}{=} \sum_v \eq(r_\lo, v)\, s_v$. For a quirky claim the partial evaluations $s_v = \hat q(r_\hi, v)$ are unchanged --- the suffix $r_\hi$ stays multilinear --- and the \emph{only} difference is that the weights $\eq(r_\lo, v)$ become a more general vector $w \in \F^{128}$, giving $\alpha = \sum_v w_v\, s_v$.

These weights have a simple closed form. Recall that the $7$ packed coordinates consist of the $\SkipParam = 6$ skipped coordinates together with one extra coordinate, so we may split the packed index as $v = (\sigma, b)$, where $\sigma \in \bitset^{\SkipParam}$ ranges over the $2^{\SkipParam}$ skip points and $b \in \bitset$ is the extra coordinate. The quirky claim reads the skipped coordinates at a univariate point $\zeta \in \F$ and the extra coordinate at a multilinear point $\rho \in \F$. The weight splits along this decomposition,
\[
  w_{(\sigma,b)} \;=\; L_\sigma(\zeta) \cdot \eq(\rho, b),
\]
where $L_\sigma$ is the Lagrange polynomial selecting the $\sigma$-th skip point (so that $L_\sigma(\zeta)$ evaluates the degree-$<2^{\SkipParam}$ univariate extension at $\zeta$), and $\eq(\rho, b)$ is the usual multilinear weight for the extra coordinate.

Every other ingredient of the reduction depends only on $r_\hi$ and fresh randomness, so we simply substitute $w_v$ for $\eq(r_\lo, v)$ and leave the rest untouched --- including soundness, which pins down each $s_v$ via the injective $\F_2$-recombination of \cref{eq:rowclaims}, regardless of $w$.

\subsection{Multilinear extension of matrix branching programs}
\label{app:mbp-mle}

Intuitively, a width $w$ matrix branching program maintains as its state a vector $s \in \F^w$. It reads its input from left-to-right and every input bit specifies a $w \times w$ linear transformation to apply to the state. At the end we take the inner product of the state with a ``sink'' vector to get the result.

More formally, a \emph{width-$w$ read-once matrix branching program (MBP)} over $\bitset^m$ specifies, for every layer $i \in [m]$ two transition matrices $M_i^{(0)},M_i^{(1)} \in \F^{w \times w}$, together with a source vector $u \in \F^w$ and sink vector $v \in \F^w$. The function $f : \bitset^m \to \F$ computed by the program is
\[
  f(x_1, \ldots, x_m) \;:=\; v^T \cdot M_m^{(x_m)} \cdot M_{m-1}^{(x_{m-1})} \cdots M_1^{(x_1)} \cdot u.
\]

As observed in \cite{HR18}, the multilinear extension of such a function can then be computed using the following identity.
\begin{equation}
\label{eq:mbp-mle}
  \hat{f}(r_1, \ldots, r_m) =
  v^T \cdot \prod_{i = m}^{1} \Bigl( (1 - r_i)\, M_i^{(0)} + r_i\, M_i^{(1)} \Bigr) \cdot u,
\end{equation}
for every $\mathbf{r} = (r_1,\dots,r_m) \in \F^m$. The identity follows by observing that by definition it holds for Boolean-valued inputs and both sides of the equation are multilinear in $\mathbf{r}$.

\section{Overview of Ligerito}\label{appx:ligerito}

\sysname uses Ligerito \cite{NA25}, instantiated with Reed--Solomon codes, as the underlying (dense) multilinear PCS. Ligerito is closely related to the WHIR PCS \cite{ArnonCFY25} but uses interleaved codes rather than Reed-Solomon codes (but the specific instantiation we focus on uses \emph{interleaved} Reed-Solomon codes).

This appendix gives a self-contained overview of Ligerito, presented as a recursive construction of an interactive oracle PCS (see \cref{sec:prelim:pcs}). Ligerito as described in \cite{NA25} works in the unique decoding regime. We give a straightforward extension to the list-decoding regime, which enables fewer queries (and hence smaller proofs).

\subsection{Commitment phase}
Our goal is to commit to the multilinear extension of a function $f \colon \bitset^{m} \to \F$, where $\F$ is a sufficiently large finite field. The PCS has two key parameters: the \emph{folding factor} $2^\ell$ and the \emph{rate} $\rho \in (0, 1)$.\footnote{Looking ahead, we remark that the construction of the PCS is recursive and these two parameters can---and will---be independently set in each level of recursion.}
We view $f$ as a $2^{\ell} \times 2^{m-\ell}$ dimensional matrix over $\F$.
Rows are indexed by $\bitset^{\ell}$, and for $\mathbf{i} \in \bitset^{\ell}$ we write $f_{\mathbf{i}}$ to denote the $\mathbf{i}$-th row of $f$.

Fix a Reed--Solomon code of rate $\rho$ over $\F$, which encodes messages of length $k \coloneqq 2^{m - \ell}$ to codewords of length $n \coloneqq 2^{m - \ell}/\rho$ via a linear map $\mathrm{Enc} \colon \F^{k} \to \F^{n}$.
The prover commits to $f$ by sending the \emph{interleaved} codeword $C$, which is the $2^\ell \times n$ matrix obtained by encoding each row of $f$, i.e., $C[\mathbf{i}, \cdot] \coloneqq \mathrm{Enc}(f_{\mathbf{i}})$ for each $\mathbf{i} \in \bitset^\ell$.

\paragraph{Binding in the list decoding regime.} How tightly this commitment binds the prover depends on the proximity radius $\gamma \in (0, 1)$ used during the opening phase, which directly translates to the number of queries issued by the verifier. In the \emph{unique decoding} regime, the proximity radius is such that the committed rows are jointly close to at most one interleaved codeword, so $C$ determines a single polynomial $\hat f$.

To reduce the query count we consider also the \emph{list decoding} regime, where $C$ is only guaranteed to be close to a small list of at most $L$ candidate codewords. To bind the prover to a single one, the verifier requests an \emph{out-of-domain} (OOD) \cite{DBLP:journals/corr/abs-1903-12243} evaluation: it samples a random point $\mathbf{z} \in \F^m$ and the prover answers with the claimed value $\hat{f}(\mathbf{z})$. Importantly, this is done immediately after the prover sends the interleaved codeword (i.e., in the commitment phase). Next, we show that, with high probability over the choice of $\mathbf{z}$, all of the candidate of the candidate codewords in the nearby list disagree on their evaluation on $\mathbf{z}$, and so this evaluation uniquely identifies one of them.

By the Schwartz-Zippel lemma, since $\mathbf{z}$ is chosen at random, the multilinears corresponding to any two distinct codewords in the list agree on $\mathbf{z}$ with probability at most $m/|\F|$. By union bounding over all pairs, the probability that there exist a pair of distinct codewords in the list whose underlying multilinears agree on $\mathbf{z}$ is at most $\binom{L}{2} \cdot \frac{m}{|\F|}$.

Thus, we need to bound the size of the list $L$. Here, we use the fact that an interleaved code inherits its distance from the underlying base code. This means that the interleaved Reed-Solomon code (which we used to encode $f$) has relative distance $1-\rho$. By the Johnson bound (see, e.g., \cite[Chapter~7]{GRS23}), for any slack parameter $\eta > 0$, at proximity radius $\gamma = 1-\sqrt{\rho}-\eta$ the committed word is close to a list of at most
\begin{equation} \label{eq:johnson-bound}
  L \leq \frac{1}{2\eta\sqrt{\rho}}
\end{equation}
codewords. Combining with the union bound, we find that the the out-of-domain evaluation fails to bind the prover with probability at most
\begin{equation*}
  \binom{L}{2}\cdot\frac{m}{|\F|} \leq \frac{1}{8\eta^2\rho}\cdot\frac{m}{|\F|} \enspace.
\end{equation*}
The slack parameter $\eta$ should be thought of as some small fixed constant (e.g., $\eta = 0.01$).

\subsection{Opening phase}
Our goal is to prove an evaluation of the form $\hat{f}(\mathbf{y}) = y$. However, to facilitate recursion (and also as it is more useful) we show how to prove more general \emph{linear evaluations} of the form
\begin{equation}\label{eq:lig-claim}
  v = \langle w, f \rangle \coloneqq \sum_{\mathbf{x} \in \bitset^m} w(\mathbf{x}) \cdot f(\mathbf{x}),
\end{equation}
where the weight $w \colon \bitset^m \to \F$ is ``MLE-friendly'': its multilinear extension $\hat w$ can be evaluated at any point in $\mathrm{poly}(m)$ time (a multilinear evaluation claim $\hat{f}(\mathbf{r})$ is the special case $w = \eq(\mathbf{y}, \cdot)$).

\paragraph{Batching claims.}
Besides the input evaluation claim $\langle w, f \rangle = v$, the verifier must also check the out-of-domain evaluation $\hat{f}(\mathbf{z}) = v'$, which can be written as $\langle w', f \rangle = v'$ for $w' = \eq(\mathbf{z},\cdot)$.
To do so, it batches them into a single evaluation claim $\langle w + \alpha \cdot w' ,f\rangle = v + \alpha \cdot v'$, where $\alpha \in \F$ is a random coefficient (this can be generalized to batch any number of evaluation claims).
Observe that if $w$ are $w'$ are MLE-friendly, then so is $w + \alpha \cdot w'$.
Thus, we can focus on proving a single evaluation claim below.

\paragraph{Sumcheck and folding.}
The prover and verifier start by running $\ell$ rounds of sumcheck on \cref{eq:lig-claim}'s claim.
The residual claim is
\begin{align}\label{eq:smaller-lig-claim}
  \sum_{\mathbf{x} \in \bitset^{m-\ell}} w'(\mathbf{x}) \cdot f'(\mathbf{x}) = v'',
\end{align}
where $v'' \in \F$ is held by the verifier and $f', w' \colon \bitset^{m - \ell} \to \F$ are obtained by partially evaluating (aka \emph{folding}) $\hat{f}, \hat{w}$ at a random point $\mathbf{s} \in \F^{\ell}$, i.e., $f'(\mathbf{x}') \coloneqq \hat{f}(\mathbf{s}, \mathbf{x})$ and $w'(\mathbf{x}') \coloneqq \hat{w}(\mathbf{s}, \mathbf{x}')$.
Observe that if $w$ is MLE-friendly, then so is $w'$.

\paragraph{Commit and recurse.} After the $\ell$ sumcheck rounds, the prover commits to $f'$ recursively, using a smaller instantiation of the IOPCS which supports $m - \ell$ variables. We can then run the smaller opening phase to prove \cref{eq:smaller-lig-claim}.

The problem is that nothing so far forces the prover to actually commit to $f'$, i.e., it may commit to an unrelated $g$ that happens to satisfy \cref{eq:smaller-lig-claim}. Hence, we must additionally check that the committed polynomial $g$ is indeed $f'$.

\paragraph{Consistency check.}
If $f'$ is not equal to $g$, then $\mathrm{Enc}(f')$ must be (coordinate-wise) far from $\mathrm{Enc}(g)$.
The verifier can therefore check consistency by sampling $t$ random coordinates $J = \{j_{1}, \dots, j_{t}\} \subset [n]$, and testing
\begin{equation*}
  \mathrm{Enc}(g)[j] = \mathrm{Enc}(f')[j] \quad \forall j \in J \enspace.
\end{equation*}
As $\mathrm{Enc}$ corresponds to a linear code, this is a set of $t$ linear evaluations of $g$; moreover, using the fact that it is a Reed--Solomon code, it can be shown that the corresponding weights are MLE-friendly~\cite{ArnonCFY25,NA25} (cf. \cite[Lemma 7.5]{BFRW25}). Thus, the prover may recursively prove these claims alongside the main claim $\langle g, w' \rangle = v''$ (via batching).

It remains to show that the verifier can obtain $\mathrm{Enc}(f')[j]$ for an arbitrary coordinate $j \in [n]$, from the original commitment matrix $C$. Let $G \in \F^{n \times k}$ be the generator matrix of the Reed--Solomon code, i.e., the matrix defining $\mathrm{Enc}(u) \coloneqq G u$, and let $G_{j}$ denote the $j$-th row of $G$.
We have the following identity for every coordinate $j$:
\begin{equation*}
  \mathrm{Enc}(f')[j] = \langle G_{j}, f' \rangle = \left\langle G_{j}, \sum_{\mathbf{i} \in \bitset^{\ell}} \eq(\mathbf{s}, \mathbf{i}) \cdot f_{i} \right\rangle = \sum_{\mathbf{i} \in \bitset^{\ell}} \eq(\mathbf{s}, \mathbf{i}) \cdot C[\mathbf{i}, j] \enspace.
\end{equation*}
Thus, the verifier can access the $j$-th coordinate of $\mathrm{Enc}(f')$ by reading the $j$-th column of $C$ and computing $\sum_{\mathbf{i}} \eq(\mathbf{s}, \mathbf{i}) \cdot C[\mathbf{i}, j]$ itself.\footnote{When compiling into a succinct argument, we arrange the columns of $C$ as leaves of a Merkle tree. Thus, opening a column only requires a single Merkle path.}

\subsection{Security} \label{sec:lig-johnson}
We sketch the opening phase's soundness analysis, with a focus on deriving precise round-by-round soundness errors of the protocol.

Let $\mathrm{RS}[\F, n, k]$ denote the Reed--Solomon code used in the PCS; its minimum distance is known to be $\delta = 1 - \rho$, where $\rho = k/n$ is the rate of the code.
Recall that a $2^{\ell}$-wise interleaved codeword is a matrix $U \in \F^{2^\ell \times n}$ where, for each $\mathbf{i} \in \bitset^{\ell}$, the row $U[\mathbf{i}, \cdot]$ is a codeword of $\mathrm{RS}[\F, n, k]$.
Let $C \in \F^{2^\ell \times n}$ be an interleaved word, and fix a proximity radius $\gamma \in (0, 1)$.
We say:
\begin{itemize}
  \item $C$ agrees with an interleaved codeword $U$ on a set of columns $A \subseteq [n]$ if, for every $\mathbf{i} \in \bitset^{\ell}$ and $j \in A$, it holds that $C[\mathbf{i}, j] = U[\mathbf{i}, j]$.
  \item $C$ agrees with the interleaved code $\mathrm{RS}[\F, n, k]^{2^{\ell}}$ on $A$ if there exists an interleaved codeword $U$ agreeing with $C$ on $A$.
  \item $C$ is $\gamma$-close to an interleaved codeword $U$ if they agree on some $A$ with $|A| \geq (1 - \gamma) \cdot n$.
\end{itemize}
We write $\Lambda_{\gamma}(C)$ to denote the list of interleaved codewords $U$ which are $\gamma$-close to $C$.
For a point $\mathbf{s}_{i} \coloneqq (s_{1}, \dots, s_{i}) \in \F^{i}$, we write $C_{\mathbf{s}_{i}}$ to denote the $\mathbf{s}_{i}$-folding of $C$, i.e., for $\mathbf{i} \in \bitset^{\ell-i}$, the entry $C_{\mathbf{s}_{i}}[\mathbf{i}, j] \coloneqq$ is defined to be $\hat{c}_{j}(\mathbf{s}_{i}, \mathbf{i})$, where $c_{j}$ denotes the $j$-th column of $C$.

\paragraph{Mutual correlated agreement.}
Our soundness analysis leverages \emph{mutual correlated agreement} (MCA) of linear codes~\cite{ArnonCFY25}.
In particular, we rely on the MCA analysis for Reed--Solomon codes from~\cite{BCHKS25}.

\begin{theorem}[MCA up to unique decoding, adapted from~{\cite[Corollary~1.4]{BCHKS25}}]\label{thm:ca-udr}
Let $\mathrm{RS}[\F, n, k]$ be a Reed--Solomon code with minimum distance $\delta = 1 - \frac{k}{n}$ satisfying $\delta \geq \frac{3 \sqrt{2}}{n}$.
Let $\gamma \in \left[\frac{\delta}{3}, \frac{\delta}{2} - \frac{3}{\delta n}\right]$. For any interleaved word $C \in \F^{2 \times n}$, it holds that
\begin{equation*}
  \Pr_{s \gets \F}\left[
    \exists A \subseteq [n], |A| \geq (1 - \gamma) \cdot n :
    \begin{array}{l}
      C_{s} \text{ agrees with } \mathrm{RS}[\F, n, k] \text{ on } A \\
      \land\; C \text{ does not agree with } \mathrm{RS}[\F, n, k]^{2} \text{ on } A
    \end{array}
  \right] \leq \frac{a}{|\F|} \enspace,
\end{equation*}
where $a = \gamma \cdot n + 1$.
\end{theorem}

\begin{theorem}[MCA up to Johnson bound, adapted from~{\cite[Theorem~4.6]{BCHKS25}}]\label{thm:ca-johnson}
Let $\mathrm{RS}[\F, n, k]$ be a Reed--Solomon code.
Denote $\rho = k/n$, the slightly reduced rate of the code.
Let $\gamma \in (0, 1- \sqrt\rho)$, $\eta \coloneqq 1- \sqrt\rho - \gamma$, and $\mu = \max\left\{\left\lceil\tfrac{\sqrt\rho}{2 \eta}\right\rceil, 3\right\}$.
For any interleaved word $C \in \F^{2 \times n}$, it holds that
\begin{equation*}
  \Pr_{s \gets \F}\left[
    \exists A \subseteq [n], |A| \geq (1 - \gamma) \cdot n :
    \begin{array}{l}
      C_{s} \text{ agrees with } \mathrm{RS}[\F, n, k] \text{ on } A \\
      \land\; C \text{ does not agree with } \mathrm{RS}[\F, n, k]^{2} \text{ on } A
    \end{array}
  \right] \leq \frac{a}{|\F|} \enspace,
\end{equation*}
where $a = \frac{2(\mu + 1/2)^5 + 3(\mu + 1/2) \gamma \rho}{3 \rho^{3/2}} \cdot n + \frac{\mu + 1/2}{\sqrt{\rho}}$.
\end{theorem}

\begin{lemma}[MCA commutes with list decoding] \label{lemma:mca-commutes}
Let $\mathrm{RS}[\F, n, k]$ be a Reed--Solomon code.
Let $\gamma \in (0, 1)$, and let $\epsilon$ be a corresponding MCA error from \cref{thm:ca-udr} or \cref{thm:ca-johnson}.
For any $\ell \in \N$ and interleaved word $C \in \F^{2^{\ell} \times n}$, it holds that
\begin{equation*}
\Pr_{s \gets \F}\left[
  \Lambda_{\gamma}(C_{s}) \neq \{U_{s} : U \in \Lambda_{\gamma}(C)\}
\right] \leq 2^{\ell - 1} \cdot \epsilon \enspace.
\end{equation*}
\end{lemma}

\begin{proof}
  Follows from~{\cite[Lemma 4.13]{ArnonCFY25}} and a union bound over the $2^{\ell-1}$ rows of $C_{s}$.
\end{proof}


\paragraph{Round-by-round soundness errors.}
Let $C$ be the interleaved word sent by the prover in the commitment phase.
Let $L_{\max}$ be the upper bound on the list size $|\Lambda_{\gamma}(C)|$ from~\cref{eq:johnson-bound} (or $L_{\max} = 1$ for the unique decoding regime), and let $\epsilon$ be the MCA error probability from \cref{thm:ca-johnson} (or~\cref{thm:ca-udr} for the unique decoding regime).

At the start of the opening phase, suppose that the prover is bound to a polynomial, i.e., there is a unique codeword in $\Lambda_{\gamma}(C)$ (the list of nearby interleaved codewords) that decodes to a polynomial agreeing with the out-of-domain evaluation claim.
Moreover, suppose this polynomial disagrees with the input evaluation claim.
It follows that every codeword in $\Lambda_{\gamma}(C)$ decodes to a polynomial disagreeing with at least one of the two evaluation claims.
We show that the verifier rejects with high probability:
\begin{itemize}
  \item \textbf{Batching claims.} Suppose that every codeword in $\Lambda_{\gamma}(C)$ disagrees with an evaluation claim.
    By soundness of batching and a union bound over the list, we find that every codeword in $\Lambda_{\gamma}(C)$ disagrees with the batched evaluation claim, except with probability $L_{\max} \cdot \frac{1}{|\F|}$.
  \item \textbf{Sumcheck.} Before the $i$-th round of the sumcheck protocol, suppose that every codeword in $\Lambda_{\gamma}(C_{\mathbf{s}_{i-1}})$ does not satisfy the $(i - 1)$-th round sumcheck claim (when $i = 1$, this is the batched evaluation claim).
    In the $i$-th round, the verifier samples $s_{i} \gets \F$.
    By soundness of sumcheck and a union bound over the list, with all but probability $L_{\max} \cdot \frac{2}{|\F|}$ it holds that the the $s_{i}$-folding of any codeword in $\Lambda_{\gamma}(C_{\mathbf{s}_{i-1}})$ disagrees with the $i$-th round sumcheck claim.
    By \cref{lemma:mca-commutes}, with all but probability $2^{\ell - i} \cdot \epsilon$ the set of $s_{i}$-foldings of codewords in $\Lambda_{\gamma}(C_{\mathbf{s}_{i-1}})$ is exactly $\Lambda_{\gamma}(C_{\mathbf{s}_{i}})$.
    Taking a union bound over both errors, we find that every codeword in $\Lambda_{\gamma}(C_{\mathbf{s}_{i}})$ disagrees with the $i$-th round sumcheck claim, except with probability $L_{\max} \cdot \frac{2}{|\F|} + 2^{\ell - i} \cdot \epsilon$.
  \item \textbf{Commit.} The prover uses a smaller instantiation of the PCS to commit to some function $g \colon \bitset^{m - \ell} \to \F$; the probability that this step fails is bounded by the binding error of the smaller PCS's commitment phase.
  \item \textbf{Consistency check.} Suppose that $g$ either does not satisfy the folded claim (i.e., the residual sumcheck claim), or its encoding $\mathrm{Enc}(g)$ is $\gamma$-far from $C_{\mathbf{s}}$.
    The former case will be handled by the next step.
    In the latter case, the verifier detects an inconsistency, i.e., $\mathrm{Enc}(g)[j] \neq C_{\mathbf{s}}[j]$ for some $j \in J$, with all but probability $(1 - \gamma)^{t}$.
    \item \textbf{Recurse.} Suppose that $g$ does not agree with the folded claim or one of the consistency check claims.
    The remaining round-by-round soundness errors follow from those of the smaller PCS's opening phase.
\end{itemize}

\begin{remark}\label{rem:query_count}
  The number of queries made by the verifier is a direct function of the proximity radius $\gamma$; in particular, one must set $t \coloneqq \frac{\lambda}{-\log(1 - \gamma)}$ in order to get $\lambda$ bits of soundness in the consistency check.
  Thus, it is beneficial to choose codes with lower rate (and hence higher $\gamma$) in the later rounds of recursion, where the prover overhead is minimal.
\end{remark}

\begin{remark}[No OOD sampling in the initial commitment] \label{remark:no-ood}
  OOD sampling is usually preferable but requires the prover to perform an additional MLE evaluation, which is actually somewhat costly for the prover.
  We observe that the initial commitment's OOD sampling can be \emph{dropped}, at the cost of the outer protocol (which uses the commitment) increasing its soundness error by a factor of the list size $L$ (security follows from a union bound over the nearby codewords).
\end{remark}

\section{Proof of the Cauchy-shift identity}
\label{app:cauchy-shift}

We prove the Cauchy-shift identity used in the Unified Lookup optimization (\cref{sec:impl:uskip}). The proof rests on two standard facts about vanishing polynomials of $\F_2$-linear subspaces:

\begin{claim}[{Linearized vanishing polynomial; see e.g.\ \cite[Ch. 3]{LN97}}]
\label{cl:linearized}
Let $S$ be an $\F_2$-linear subspace of $\F_{2^\tau}$ and set $Z_S(x) := \prod_{s \in S}(x + s)$. Then:
\begin{enumerate}
\item $Z_S$ is \emph{additive}: $Z_S(x + y) = Z_S(x) + Z_S(y)$ for all $x, y \in \F_{2^\tau}$.
\item The formal derivative $Z_S'$ is constant on $S$: $Z_S'(s) = D$ for every $s \in S$, where $D := \prod_{s' \in S \setminus \{0\}} s'$.
\end{enumerate}
\end{claim}

\begin{lemma}[Translation-invariance of the LDE matrix]
\label{lem:cauchy-shift}
Let $\tau$ be a power of two, let $S \subseteq \F_{2^\tau}$ be an $\F_2$-linear subspace with $|S| \ge \tau$, let $\delta \in \F_{2^\tau} \setminus S$, and let $\Lambda := \delta + S$. Fix any $\F_2$-basis of $S$ and use it to index $S$ (resp.\ $\Lambda$) by $[|S|]$ via the natural $\F_2$-linear bijection (resp.\ its $\delta$-translate). Let $M := \mathrm{NTT}_\Lambda \circ \mathrm{iNTT}_S$. Then for every $i \in [|S|]$, $b \in \{0, 1, \ldots, |S|/\tau - 1\}$, and $j \in \{0, 1, \ldots, \tau - 1\}$,
\[
  M\bigl[i,\, \tau b + j\bigr] \;=\; M\bigl[\,i \oplus \tau b,\, j\,\bigr].
\]
\end{lemma}

\begin{proof}
Write $\{e_0, \ldots, e_{k_{\mathrm{s}}-1}\}$ for the chosen basis, and for $k \in [|S|]$ with binary expansion $k = \sum_r k_r 2^r$, let $s_k := \sum_{r:k_r=1} e_r \in S$ and $\lambda_i := \delta + s_i \in \Lambda$ denote the corresponding field elements. By Lagrange interpolation on $S$ (using $a - b = a + b$ in characteristic $2$),
\[
  M[i, k] \;=\; \frac{Z_S(\lambda_i)}{(\lambda_i + s_k) \cdot Z_S'(s_k)},
\]
where $Z_S(x) := \prod_{s \in S}(x + s)$ is the vanishing polynomial of $S$. By \cref{cl:linearized}(1), $Z_S(\lambda_i) = Z_S(\delta + s_i) = Z_S(\delta) + Z_S(s_i) = Z_S(\delta)$ for every $i$ (the second term vanishes as $s_i \in S$). By \cref{cl:linearized}(2), $Z_S'(s_k) = D$ for every $k$. Setting $\mu := Z_S(\delta)$, we obtain
\begin{equation}
\label{eq:M-closed-form}
  M[i, k] \;=\; \frac{\mu}{(\lambda_i + s_k) \cdot D}.
\end{equation}

Finally, $\tau$ being a power of $2$ with $\tau \le |S|$ gives $(\tau b) + j = (\tau b) \oplus j$ (no carry), and the $\F_2$-basis indexing gives $s_{(\tau b) \oplus j} = s_{\tau b} + s_j$ and $\lambda_{i \oplus (\tau b)} = \lambda_i + s_{\tau b}$. Substituting into~\eqref{eq:M-closed-form},
\begin{align*}
  M[i, (\tau b) + j]
  &= \frac{\mu}{(\lambda_i + s_{\tau b} + s_j) \cdot D} \\
  &= \frac{\mu}{(\lambda_{i \oplus (\tau b)} + s_j) \cdot D}
  \;=\; M[i \oplus (\tau b), j]. \qedhere
\end{align*}
\end{proof}

\end{document}